\documentclass[graybox]{svmult}

\usepackage{type1cm}        
\usepackage{makeidx}         
\usepackage{graphicx}        
\usepackage{multicol}        
\usepackage[bottom]{footmisc}

\usepackage{newtxtext}       %
\usepackage[varvw]{newtxmath}       

\usepackage{balance}
\usepackage{color}
\usepackage{caption}
\usepackage{enumerate}
\usepackage{cite}
\usepackage{empheq}
\usepackage{mathrsfs}

\usepackage{mathtools}

\usepackage{tikz}
\usepackage{circuitikz}

\usepackage[font=small,labelfont=bf]{caption}

\usepackage{subfig}

\makeatletter
\let\proof\@undefined                        
\let\endproof\@undefined                  
\makeatother
\usepackage{graphicx,amssymb,amstext,amsmath,amsthm}
\renewcommand\qedsymbol{\small$\blacksquare$}
\usepackage[bookmarks=true]{hyperref}
\usepackage{algorithm,algorithmicx,algpseudocode}
\algnewcommand{\algorithmicgoto}{\textbf{go to}}%
\algnewcommand{\Goto}[1]{\algorithmicgoto~\ref{#1}}%
\algnewcommand{\LineComment}[1]{\Statex \(\triangleright\) #1}
\algnewcommand{\LineCommentN}[1]{\Statex \hspace{1cm}\(\triangleright\) #1}

\usepackage{multirow}
\usepackage{stfloats}

\newcommand{\diag}{\mbox{diag}}

	\newtheorem{assumption}{Assumption}
\newtheorem{lem}{Lemma}

\newtheorem*{proof}{Proof}

\newcommand{\mk}[1]{{\color{black}#1}}
\newcommand{\tpa}[1]{{\color{black}#1}}

\makeindex             

\begin{document}

\title*{Optimal Interval Observers for Bounded Jacobian Nonlinear Dynamical Systems}
\author{Mohammad Khajenejad\orcidID{0000-0002-1684-8167},\\ Tarun Pati\orcidID{0000-0003-3237-8828} and \\
Sze Zheng Yong\orcidID{0000-0002-2104-3128}}
\institute{Mohammad Khajenejad \at The University of Tulsa, Tulsa, OK, USA, \email{mohammad-khajenejad@utulsa.edu}
\and Tarun Pati \at Northeastern University, Boston, MA, USA \email{pati.ta@northeastern.edu}
\and Sze Zheng Yong \at Northeastern University, Boston, MA, USA \email{s.yong@northeastern.edu}}
%
%
\maketitle


\abstract{In this chapter, we introduce two interval observer designs for discrete-time (DT) and continuous-time (CT) nonlinear systems with bounded Jacobians that are affected by bounded  uncertainties. Our proposed methods utilize the concepts of mixed-monotone decomposition and embedding systems to design correct-by-construction interval framers, i.e., the interval framers inherently bound the true state of the system without needing any additional constraints. Further, our methods leverage techniques for positive/cooperative systems to guarantee global uniform ultimate boundedness of the framer error, i.e., the proposed interval observer is input-to-state stable. Specifically, our two interval observer designs minimize the $\mathcal{H}_{\infty}$ and $L_1$ gains, respectively, of the associated linear comparison system of the framer error dynamics. Moreover, our designs adopt a \mk{multiple-gain} observer structure, which offers additional degrees of freedom, along with coordinate transformations that may improve the feasibility of the resulting optimization programs.  We will also discuss and propose computationally tractable 
optimization formulations to compute the observer gains. Finally, we compare the efficacy of the proposed designs against existing DT and CT interval observers.}
\section{Introduction}
State estimation is of great importance across a spectrum of cyber-physical systems and engineering domains, from autonomous vehicles to power systems, serving as the bedrock for monitoring, decision-making, and control processes. Among the diverse array of state estimation methods, interval observers emerge as a compelling solution for systems affected by uncertainty. Interval observers furnish interval-valued state estimates, i.e., upper and lower bounds of the states, offering a unique advantage in scenarios where uncertainties manifest as sets rather than point values. This is particularly pertinent when faced with non-stochastic uncertainties or when the distributions governing these uncertainties are unknown. Thus, this chapter considers the interval observer design problem for nonlinear bounded-error continuous-time (CT) and discrete-time (DT) dynamical systems using a unified framework.
\subsection{Literature Review}
A vast body of scholarly literature delves into the intricacies of designing set-valued/interval observers for a wide spectrum of systems, including linear, nonlinear, cooperative/monotone, mixed-monotone, distributed, hybrid and Metzler dynamics~ \cite{chebotarev2015interval,wang2015interval, tahir2021synthesis,khajenejad2021intervalACC,khajenejad2020simultaneousCDC,farina2000positive,khajenejad2022resilientspringer,pati2023intervalhybridcdc}. Across this literature, a prevailing approach involves synthesizing observer gains to simultaneously ensure the Schur/Hurwitz stability and cooperativity of observer error dynamics, which, while conceptually sound, often leads to significant theoretical and computational challenges. To address these complexities, researchers have explored diverse strategies, such as harnessing interval arithmetic-based methods~\cite{kieffer2006guaranteed}, transforming systems into positive forms~\cite{cacace2014new}, and employing time-varying/invariant state transformations, e.g.,~\cite{tahir2021synthesis}. 
However, for general classes of nonlinear systems, the task remains formidable, necessitating the utilization of bounding mappings/decomposition functions~\cite{khajenejad2021tight} to reformulate observer design problems as semidefinite programs (SDPs) or optimizations with linear matrix inequality (LMI) constraints~\cite{wang2012observer,wang2015interval,efimov2016design,briat2016interval,tahir2021synthesis,moh2022intervalACC,khajenejad2022simultaneousIJRNC,khajenejad2023distributedACC,khajenejad2023distributedresilientACC}. Nevertheless, the obtained matrix inequalities can still impose significant restrictions, potentially rendering solutions infeasible for certain systems due to the imposition of various conditions and upper bounds. In response to this challenge, numerous coordinate transformations have been suggested to alleviate design constraints and enable the computation of feasible observer gains~\cite{dinh2014interval,mazenc2011interval}. However, unfortunately, existing approaches are unable to concurrently synthesize or design both the coordinate transformation and observer gains. 

Moreover, recent advancements in uncertain systems have prompted considerations of additional design criteria, such as robustness against noise and uncertainty. This has led to the translation of design problems into SDPs encompassing a combination of framer satisfaction (i.e., correctness property), observer stability, and noise attenuation constraints, albeit at the cost of potential conservatism and increased computational burden~\cite{tahir2021synthesis,cacace2014new,efimov2016design,khajenejad2021simultaneousECC}. In prior research endeavors, an $\mathcal{H}_{\infty}$-optimal observer was proposed by the authors in \cite{khajenejad_H_inf_2022} to mitigate these issues by minimizing the $L_2$-gain of the observer error system without coordinate transformation for both CT and DT systems, which involves solving of an SDP. On the other hand, an interval observer design with additional degrees of freedom (i.e., observer gains) was introduced in~\cite{pati2022L} for both CT and DT systems to minimize the $L_1$-gain of the error system, which instead involves solving of a mixed-integer linear program (MILP). Building upon these foundations, the present chapter introduces a unified framework for designing $\mathcal{H}_{\infty}$ or $L_1$ robust interval observers for both CT and DT bounded-error nonlinear systems with bounded Jacobians. 
\section{Preliminaries}
 {\emph{{Notation}.}} $ \mathbb{R}^n$, $\mathbb{R}^n_{>0}$, $\mathbb{R}^{n  \times p}$,  $\mathbb{N}_n$ and $\mathbb{N}$ denote the $n$-dimensional Euclidean space, positive vectors of size $n$, matrices of size $n$ by $p$, natural numbers up to $n$ and natural numbers, respectively. For a vector $v \in \mathbb{R}^n$, its vector $p$-norm is given by $\|v\|{_{p}\triangleq (\sum_{i=1}^n {|v_i|^p}})^{\frac{1}{p}}$ \tpa{and $\diag(v)$ represents a diagonal matrix whose diagonal elements are elements of $v$}, while for a matrix $M\in \mathbb{R}^{n  \times p}$, $M_{ij}$ represents its $j$-th column and $i$-th row entry,  $\textstyle{\mathrm{sgn}}(M)$ represents its element-wise signum function,   $M^{\oplus}\triangleq \max(M,\mathbf{0}_{n\times p})$, {$M^{\ominus}\triangleq M^{\oplus}-M$, and $|M|\triangleq M^{\oplus}+M^{\ominus}$} is its element-wise absolute value. Moreover, $M^\text{d}$ is a diagonal matrix with (only) the diagonal elements of the square matrix $M\in \mathbb{R}^{n  \times n}$, $M^\text{nd} \triangleq M-M^\text{d}$ is the matrix with only its off-diagonal elements, and $M^{\text{m}} \triangleq M^\text{d}+|M^\text{nd}|$ is the ``Metzlerized" matrix\footnote{A Metzler matrix is a square matrix in which all the off-diagonal components are nonnegative (equal to or greater than zero).}. Further, $M \succ 0$ and $M \prec 0$ (or $M \succeq 0$ and $M \preceq 0$) denote that $M$ is positive and negative   (semi-)definite, respectively, while all vector and matrix inequalities are element-wise inequalities, and the matrices of zeros and ones of dimension $n \times p$ are denoted as $\mathbf{0}_{n \times p}$ and $\mathbf{1}_{n \times p}$, respectively. In addition, a \mk{continuous} function \mk{$\alpha: [0,a) \to \mathbb{R}_+$} is \mk{said to belong to} class $\mathcal{K}$ if it is strictly increasing \mk{and $\alpha(x)=0$}, \mk{(and consequently it is  positive definite, i.e., $\alpha(x)=0$ for $x=0$ and $\alpha(x)>0$ otherwise)}. \mk{Moreover, the function $\alpha$ belongs to} class $\mathcal{K}_{\infty}$ if \mk{it belongs to class $\mathcal{K}$, $a=\infty$, and $\lim_{r \to \infty}\alpha(r)=\infty$, i.e., $\alpha$} is unbounded. \mk{Finally, a continuous function} $\lambda : \mk{[0,a) \times [0,\infty)} \to \mathbb{R}_+$ is \mk{said to belong to} class $\mathcal{KL}$ if for each fixed $t\geq 0$, \mk{the function} $\lambda(s,t)$ \mk{belongs to} class $\mathcal{K}$, for each fixed $s \geq 0$, $\lambda(s,t)$ \mk{is decreasing in $t$ and $\lim_{t \to \infty}\lambda(s,t)=0$.} 
\begin{definition}[Interval]\label{defn:interval}
An (n-dimensional) interval denoted as {$\mathcal{I} \triangleq [\underline{z},\overline{z}]  \subset 
\mathbb{R}^n$} is the set of vectors $z \in \mathbb{R}^{n_z}$ that satisfy $\underline{{z}} \le {z} \le \overline{{z}}$.
A matrix interval follows a similar definition.
\end{definition}
\begin{definition}[Jacobian Sign-Stability]\label{def:JSS}  
A vector-valued function $g :{\mathcal{Z}} \subset \mathbb{R}^{n_z} \to  \mathbb{R}^{p}$ is Jacobian sign-stable (JSS), if in its domain $\mathcal{Z}$, the entries of its Jacobian matrix {do} not change signs, i.e., if one of the following hold: 
\begin{align*}
 J^g_{ij}({z}) \geq 0 \ \text{ or } J^g_{ij}({z}) \leq 0 
\end{align*} 
for all ${z \in \mathcal{Z}}, \forall i \in \mathbb{N}_p,\forall j \in \mathbb{N}_{n_z}$,
where $J^g({z})$ represents the Jacobian matrix of the mapping $g$ evaluated at ${z \in \mathcal{Z}}$.
\end{definition} \label{defn:JSS}
\begin{proposition}[Jacobian Sign-Stable Decomposition] \cite[Proposition 2]{moh2022intervalACC}\label{prop:JSS_decomp}
For a mapping $g :{\mathcal{Z}} \subset \mathbb{R}^{n_z} \to  \mathbb{R}^{p}$, if $ J^g({z}) \in [\underline{J}^g,\overline{J}^g]$ for all ${z \in \mathcal{Z}}$, where $\underline{J}^g,\overline{J}^g \in \mathbb{R}^{p \times n_z}$ are known matrices, then the function $g$ can be decomposed as the sum of a JSS mapping $\mu$ and \mk{a remainder} affine mapping $H{z}$ (that is also JSS), in an additive \emph{remainder-form}: 
\begin{align}\label{eq:JSS_decomp}
\forall z \in \mathcal{Z},g(z)=Hz+\mu(z),
\end{align}
 where the matrix $H\in\mathbb{R}^{p \times n_z}$, satisfies
 \begin{align}\label{eq:H_decomp}
 \forall (i,j) \in \mathbb{N}_p \times \mathbb{N}_{n_z}, H_{ij}=\underline{J}^g_{ij} \ \lor H_{ij}=\overline{J}^g_{i,j} .    
 \end{align}
\end{proposition}
\mk{
\begin{remark}
As can be seen from \eqref{eq:JSS_decomp}, for  an $n$-dimensional vector field (i.e., $n_z=p=n$), we have $2^{(n^2)}$ choices for $H$, all of which lead to correct JSS decompositions and can be applied. The investigation of the best choice of $H$ for interval observer designs is a subject of future work. In this chapter and to implement the observer on the simulation examples, we selected $H$ similar to \cite{yang2019sufficient}, which we showed in our previous work \cite{khajenejad2021tight} that it belongs to the set of choices provided in Proposition 1.
\end{remark}
}
\begin{definition}[Mixed-Monotonicity {and} Decomposition Functions] \cite[Definition 1]{abate2020tight},\cite[Definition 4]{yang2019sufficient} \label{defn:dec_func}
Consider the uncertain discrete-time (DT) or continuous-time (CT) dynamical system {with initial state $x_0 \in \mathcal{X}_0 \triangleq [\underline{x}_0,\overline{x}_0] {\subset \mathbb{R}^{n}}$} and an augmentation of known inputs and process noise $w_t \in \mathcal{W} \triangleq [\underline{w},\overline{w}] \subset \mathbb{R}^{n_w} $:
\begin{align}\label{eq:mix_mon_def}
x_t^+= g(z_t) \triangleq g(x_t,{w_t}), \quad z_t \triangleq [x^\top_t \ w^\top_t]^\top,
\end{align}
Moreover, ${g}:{\mathcal{Z}} \subset \mathbb{R}^{n_z} \to \mathbb{R}^{n}$ is the vector field 
with augmented state ${z_t} \in \mathcal{Z} \triangleq \mathcal{X} \times \mathcal{W} \subset \mathbb{R}^{n_z}$ as its domain, where $\mathcal{X}$ is the entire state space and  $n_z=n+n_w$. 
\begin{itemize}
\item If \eqref{eq:mix_mon_def} is a DT system, i.e., $x_t^+ \triangleq x_{t+1}$, then, a mapping $g_d:\mathcal{Z}\times \mathcal{Z} \to \mathbb{R}^{p}$ is 
a {DT mixed-monotone} decomposition function with respect to $g$, if 
\begin{enumerate}[(i)]
    \item $g_d(z,z)=g(z)$,
    \item $g_d$ is monotone increasing in its first argument, i.e., 
    $$\hat{z}\ge z \Rightarrow g_d(\hat{z},z') \geq g_d(z,z'),$$
    \item $g_d$ is monotone decreasing in its second argument, i.e., 
    $$\hat{z}\ge z \Rightarrow g_d(z',\hat{z}) \leq g_d(z',z).$$ 
\end{enumerate}
\item On the other hand, if \eqref{eq:mix_mon_def} is a CT system, i.e., if $x_t^+ \triangleq \dot{x}_t$, then a mapping $g_d:\mathcal{Z}\times \mathcal{Z} \to \mathbb{R}^{p}$ is a CT mixed-monotone decomposition function with respect to $g$, 
if 
\begin{enumerate}[(i)]
\item $g_d(z,z)=g(z)$,
 \item $g_d$ is monotone increasing in its first 
argument only with respect to ``off-diagonal'' arguments, i.e., 
$$\forall (i,j) \in \mathbb{N}_{n} \times \mathbb{N}_{n_z} \land i \ne j,\hat{z}_j\ge z_j, \hat{z}_i= z_i  \Rightarrow g_{d,i}(\hat{z},z') \geq g_{d,i}(z,z'),$$
    \item $g_d$ is monotone decreasing in its second argument, i.e., 
    $$\hat{z}\ge z \Rightarrow g_d(z',\hat{z}) \leq g_d(z',z).$$ 
\end{enumerate}
\end{itemize}
\end{definition}
\begin{proposition}[Tight Decomposition Functions for Linear Systems]\label{prop:lin_decomposition}
Suppose the system in \eqref{eq:mix_mon_def} is linear time-invariant, i.e., 
\begin{align}\label{eq:mix_mon_lin}
x_t^+=g(x_t,w_t)= A_gx_t+B_gw_t.
\end{align}
Then, the mapping $g$ admits the following decomposition function: 
\begin{align}\label{eq:gd}
g_d(x_1,w_1,x_2,w_2)= A_g^\uparrow x_1-A_g^\downarrow x_2+B_g^\oplus w_1-B_g^\ominus w_2, 
\end{align}
where
\begin{itemize}
\item $A_g^\uparrow \triangleq A_g^\oplus$ and $A_g^\downarrow \triangleq A_g^\ominus$ if \eqref{eq:mix_mon_lin} is a DT system, and 
\item $A_g^\uparrow \triangleq A_g^{\text{\emph{nd}},\oplus}+A_g^{\text{\emph{d}}}$ and $A_g^\downarrow \triangleq A_g^{\text{\emph{nd}},\ominus}$ if \eqref{eq:mix_mon_lin} is a CT system.
\end{itemize}
\end{proposition}
\begin{proof}
Defining $z\triangleq (x,w)$, it is straightforward to see from \eqref{eq:gd} that 
\begin{align*} 
g_d(z,z)&=(A_g^\uparrow-A_g^{\downarrow})x+(B_g^\oplus-B_g^\ominus)w\\
&=\begin{cases} (A_g^\oplus-A_g^{\ominus})x+B_gw, \ \quad \quad \quad \quad \quad \quad \quad \quad \quad \quad \quad \quad \quad \text{(DT)}\\
(A_g^{\text{\emph{nd}},\oplus}-A_g^{\text{\emph{nd}},\ominus}+A_g^{\text{\emph{d}}})x+B_gw=(A_g^{\text{\emph{nd}}}+A_g^{\text{\emph{d}}})x+B_gw, \ \text{(CT)}\end{cases}\\
&=A_gx+B_gw =g(z). 
\end{align*}
This proves condition (i) in Definition \ref{defn:dec_func} for both DT and CT cases.

Moreover, in the DT case, let $\hat z=(\hat x,\hat w) \geq z=(x,w) \Leftrightarrow \hat x \geq x, \hat w \geq w$. Then, for any $z'=(x',w') \in \mathcal{Z}$, the following holds $$g_d(\hat z,z')=A_g^\oplus \hat x -A_g^\ominus x'+B_g^\oplus \hat w-B_g^\ominus x' \geq A_g^\oplus  x -A_g^\ominus x'+B_g^\oplus  w-B_g^\ominus x'=g_d(z,z'),$$ where the inequality follows from the non-negativity of $A_g^\oplus$ and $B_g^\oplus$, which preserves the order of any two ordered multiplicands, and in particular, $\hat x,x$ and $\hat w, w$. By similar arguments, $g_d(z', \hat z) \leq g_d(z',z)$. In other words, conditions (ii) and (iii) for the DT case in Definition \ref{defn:dec_func} hold.

Furthermore, to prove condition (ii) for the CT case, we consider an arbitrary dimension $i \in \mathbb{N}_n$ and an ordered $\hat z,z \in \mathcal {Z}$ such that $\hat z_j \geq z_j , \forall j \in \mathbb{N}_{n_z},  j\ne i$ and $\hat z_i=\hat z_i$. Then, for any $z' \in \mathcal{Z}$ we have 
\begin{align*}
g_{d,i}(\hat z,z')&=(A_{g,i}^{\text{\emph{nd}},\oplus}+A_{g,i}^{\text{\emph{d}}})\hat x-A_{g,i}^{\text{\emph{nd}},\ominus}x'+B_{g,i}^\oplus \hat w-B_{g,i}^\ominus w'\\ 
&\geq (A_{g,i}^{\text{\emph{nd}},\oplus}+A_{g,i}^{\text{\emph{d}}}) x-A_{g,i}^{\text{\emph{nd}},\ominus}x'+B_{g,i}^\oplus  w-B_{g,i}^\ominus w'g_{d,i}(z,z'),
\end{align*}
where the inequality follows from the non-negativity of $A_{g,i}^{\text{\emph{nd}},\oplus}$ and $B_{g,i}^\oplus$, and the facts that $\hat x_j \geq x_j, \forall j \neq i$, $\hat x_i=x_i$, and $\hat w \geq w$. This proves condition (ii) for the CT case in Definition \ref{defn:dec_func}. Finally, the proof of (iii) for the CT case is similar to its counterpart for the DT case. 
\end{proof}
\begin{proposition}[Tight and Tractable DT Decomposition Functions for Nonlinear JSS Mappings {\cite[Proposition 4]{moh2022intervalACC} \& \cite[Lemma 3]{moh2022intervalACC}}]\label{prop:tight_decomp}
Let $\mu:\mathcal{Z} \subset \mathbb{R}^{n_z} \to \mathbb{R}^p$ be a JSS mapping on its domain, where $J^{\mu} \in [\underline{J}^{\mu},\overline{J}^{\mu}]$. Then, it admits a tight DT decomposition function for each $\mu_i,\ i \in \mathbb{N}_p$, as follows: 
\begin{align}\label{eq:JJ_decomp}
\mu_{d,i}(z_1,z_2)=\mu_i(D^iz_1+(I_{n_z}-D^i)z_2), 
\end{align}
for any ordered $z_1, z_2 \in \mathcal{Z}$, where $D^i$ is a binary diagonal matrix determined by which vertex of the interval $[{z}_2,{z}_1]$ or $[z_1,z_2]$ that maximizes (if $z_2 \leq z_1$) or minimizes (if $z_2 > z_1$) the function 
$\mu_i$, and can be found in closed-form as: 
\begin{align}\label{eq:Dj}
D^i=\textstyle{\mathrm{diag}}(\max(\textstyle{\mathrm{sgn}}(\overline{J}^{\mu}_i),\mathbf{0}_{1,n_z})),
\end{align}
where $\mathrm{diag}(v)$ is a diagonal matrix whose diagonal elements are elements of $v$. 
Moreover, for any interval domain $\underline{z} \leq z \leq \overline{z}$, with $\underline{z},z,\overline{z} \in \mathcal{Z}$, the following inequality holds:
\begin{align}\label{eq:JSS_up_bound}
\delta^{\mu}_z \triangleq \mu_d(\overline{z},\underline{z})-\mu_d(\underline{z},\overline{z}) \leq F_{\mu} (\overline{z}-\underline{z}), 
\end{align}
where $F_{\mu}=(\overline{J}^{\mu})^{\oplus}+(\underline{J}^{\mu})^{\ominus}$.
\end{proposition}
\begin{definition}[One-Sided Decomposition Functions]\cite[Definition 2]{khajenejad2021tight}\label{def:one_sided_dec}
Consider system \eqref{eq:mix_mon_def} and suppose there exist mixed-monotone mappings $\overline{g}_d,\underline{g}_d:\mathcal{Z} \times \mathcal{Z} \to \mathcal{R}^n$ such that for any $z, \underline{z},\overline{z} \in \mathcal{Z}$ the following statement holds: 
$$\underline{z} \leq z \leq \overline{z} \Rightarrow \underline{g}_d(\underline{z},\overline{z}) \leq g(z) 
\leq \overline{g}_d(\overline{z},\underline{z}),$$
where with a slight abuse of notation, we overload the notation
of $\underline{g}_d(\underline{z},\overline{z})$ and $\overline{g}_d(\overline{z},\underline{z})$ when \eqref{eq:mix_mon_def} is a CT system to represent the case with a fixed $z_i$ for the $i^{th}$ function $f_i$. Then, $\overline{g}_d$ and $\underline{g}_d$ are called upper and lower one-sided decomposition functions for $g$ over $[\underline{z},\overline{z}]$, respectively.
\end{definition}
\begin{definition}[Generalized Embedding System]\cite[Definition 3]{khajenejad2021tight}\label{def:embedding}  
For an $n$-dimensional uncertain system \eqref{eq:mix_mon_def} with any pair of one-sided decomposition function $\overline{g}_d,\underline{g}_d$, its generalized embedding system is defined as the following $2n$-dimensional certain dynamical system with initial condition 
$\begin{bmatrix} \overline{x}_0^\top & \underline{x}_0^\top\end{bmatrix}^\top$:
\begin{align} \label{eq:embedding}
\begin{bmatrix}{\underline{x}}_t^+ \\ {\overline{x}}_t^+ \end{bmatrix}=\begin{bmatrix}  \underline{g}_d(\begin{bmatrix}(\underline{x}_t)^\top \, \underline{w}^\top \end{bmatrix}^\top,\begin{bmatrix}(\overline{x}_t)^\top \, \overline{w}^\top\end{bmatrix}^\top) \\  \overline{g}_d(\begin{bmatrix}(\overline{x}_t)^\top \, \overline{w}^\top\end{bmatrix}^\top,\begin{bmatrix}(\underline{x}_t)^\top \, \underline{w}^\top \end{bmatrix}^\top) \end{bmatrix}. 
\end{align}
\end{definition}
\begin{proposition}[Framer Property]\cite[Proposition 3]{khajenejad2021tight}\label{prop:frame_property}
The solution to a generalized embedding system \eqref{eq:embedding} with one-sided decomposition functions $\overline{g}_d,\underline{g}_d$ corresponding to the system dynamics in \eqref{eq:mix_mon_def} has a \emph{state framer property}, i.e., it is guaranteed to frame the unknown state trajectory $x_t$ of \eqref{eq:mix_mon_def}: $$\underline{x}_t\le x_t\le \overline{x}_t, \forall t\in\mathbb{T}.$$
\end{proposition}
\section{Problem Formulation} 
\label{sec:Problem}

\noindent\textbf{\emph{System Dynamics.}} We consider a class of uncertain/noisy discrete-time (DT) or continuous-time (CT) nonlinear systems given as:
\begin{align} \label{eq:system}
\begin{array}{ll}
\mathcal{G}: \begin{cases} {x}_t^+ ={f}(x_t)+Bu_t+Ww_t ,   \\
\  y_t = {h}(x_t)+Du _t+Vv_t, 
\end{cases} \ \text{for all} \ t \in {\mathbb{T}},  
\end{array}
\end{align}
where
\begin{itemize}
\item $x_t^+=x_{t+1}$ and ${\mathbb{T}}= \{0\}\cup \mathbb{N}$, if $\mathcal{G}$ is a DT system, and 
\item $x_t^+=\dot{x}_t$ and $\mathbb{T} = \mathbb{R}_{\ge 0}$, if $\mathcal{G}$ is a CT system. 
\end{itemize}
Moreover, $x_t \in \mathcal{X} \subset \mathbb{R}^n$, {$w_t \in \mathcal{W} \triangleq [\underline{w},\overline{w}] \subset \mathbb{R}^{n_w},v_t \in \mathcal{V} \triangleq [\underline{v},\overline{v}] \subset \mathbb{R}^{n_v}$}, $u_t \in \mathcal{U} \triangleq [\underline{u},\overline{u}] \subset \mathbb{R}^m$ and $y_t \in \mathbb{R}^l$ are state, process noise, measurement noise, known control input and output measurement signals, respectively.
Further, $f:\mathbb{R}^n  \to \mathbb{R}^n$ denotes the nonlinear state vector field, while $h:\mathbb{R}^n  \to \mathbb{R}^l$ is a mapping that can represent system observation/measurements and/or state constraints. Additionally, $W \in \mathbb{R}^{n \times n_w}, B \in \mathbb{R}^{n \times m},V \in \mathbb{R}^{l \times n_v}$ and $\mk{D} \in \mathbb{R}^{l \times m}$ are known matrices. Furthermore, we assume the following: 
\begin{assumption}[Known Disturbance Bounds, Measurements, \& Inputs] 
\label{ass:known_input_output}
The disturbance/noise and input bounds $ \underline{w},\overline{w}$, $ \underline{v},\overline{v}$ and $ \underline{u},\overline{u}$, as well as the signals $y_t$ (output) and $u_t$ (input, if any) are known at all times. Moreover, the initial state $x_0$ is such that $x_0 \in \mathcal{X}_0 = [ \underline{x}_0,\overline{x}_0]$ {with} known 
bounds $\underline{x}_0$ and $\overline{x}_0$.
\end{assumption}

\begin{assumption}[Bounded Jacobian Nonlinearities]
\label{ass:mixed_monotonicity}
The mappings/functions $f$ and $h$ are known and differentiable, and have bounded Jacobians in their domain\footnote{The differentiability assumption is primarily for ease of exposition and can be relaxed to a weaker assumption (cf. \cite{khajenejad2021tight} for more details).}. Moreover, the lower and upper bounds of their Jacobian matrices over the entire state space $\mathcal{X}$, i.e., $\underline{J}^{f},\overline{J}^{f} \in \mathbb{R}^{n \times {n}}$ and $\underline{J}^{h},\overline{J}^{h} \in \mathbb{R}^{l \times n}$ are known.
\end{assumption}

We aim to estimate the state trajectories of the plant $\mathcal{G}$ in \eqref{eq:system}. To pose this problem formally, we first define the notions of correctness, input-to-state stability (ISS), and optimal interval observer, as follows.
\begin{definition}[Correct Interval Framers and Framer Errors]\label{defn:framers}  
The signals/sequences $\overline{x},\underline{x}: {\mathbb{T}} \to \mathbb{R}^n$ are called upper and lower framers for the system states $x_t$ of the nonlinear plant $\mathcal{G}$ in \eqref{eq:system} if 
\begin{align}\label{eq:correctness}
\underline{x}_t \leq x_t \leq \overline{x}_t, \ \forall t \in {\mathbb{T}}, {\forall w_t \in \mathcal{W},\forall v_t \in \mathcal{V}}.
\end{align}
Further, $e_t \triangleq \overline{x}_t-\underline{x}_t$ is called the \emph{framer error} at time $t$. Any dynamical system whose states are correct framers for the system states of the plant $\mathcal{G}$, i.e., with $e_t \ge 0, \forall t\in\mathbb{T}$, is called a \emph{correct} interval framer for system \eqref{eq:system}. 
\end{definition}
\begin{definition}
An interval framer is input-to-state stable (ISS), if the frame error (cf. Definition \ref{defn:framers}) is bounded as follows: 
\begin{align}\label{eq:L1-ISS}
 \|e_t\|_2 \leq \beta(\|e_0\|_2,t)+\rho(\|\delta\|_{\ell_\infty}), \forall t \in \mathbb{T},
\end{align} 
where $\beta$ and $\rho$ are functions of classes $\mathcal{KL}$ and $\mathcal{K}_{\infty}$, respectively, with the $\ell_\infty$ signal norm $\|\delta\|_{\ell_\infty}\triangleq \sup_{t\in [0,\infty)} \|\delta_t\|_{2} =\|\delta\|_2$. Moreover, $\varepsilon_t$ and $\delta_t=\delta\triangleq \begin{bmatrix} \delta_w^\top & \delta_v^\top \end{bmatrix}^\top$ are the framer error and combined noise signals, with $\delta_w \triangleq\overline{w}-\underline{w}$ and $\delta_v\triangleq\overline{v}-\underline{v}$.
\end{definition}
\begin{definition}[$L_1$ or $\mathcal{H}_{\infty}$-Robust \& Optimal Interval Observer]\label{defn:L_1} 
An interval framer  $\hat{\mathcal{G}}$ is $L_1$-robust or $\mathcal{H}_{\infty}$-robust and optimal, if the $L_1$-gain (with $s=\ell_1$) or $\mathcal{H}_{\infty}$-gain (with $s=\ell_2$), respectively, of the framer error system $\tilde{\mathcal{G}}$, defined below, is minimized:
\begin{align}\label{eq:L1_Def}
\|\tilde{\mathcal{G}}\|_{s} \triangleq \sup_{\|\delta\|_{s}=1} \|e\|_{s}, \ s \in \{\ell_1,\ell_2\},
\end{align}
where $\|\nu\|_{\ell_1} \triangleq \int_0^\infty \|v_t\|_1 dt$ and $\|\nu\|_{\ell_2} \triangleq \int_0^\infty \sqrt{\|\nu_t\|^2_2} dt$ are the $\ell_1$ and $\ell_2$ signal norms for $\nu \in \{e,\delta\}$, respectively. 
\end{definition}
The optimal observer design problem can be stated as follows:
\begin{problem}[ISS \& $L_1$ or $\mathcal{H}_{\infty}$-Optimal Interval Observer Design]
\label{prob:SISIO}
Given the nonlinear system in \eqref{eq:system}, 
design a correct and $L_1$ or $\mathcal{H}_{\infty}$-robust optimal interval observer (cf. Definitions \ref{defn:L_1}) whose framer error (cf. Definition \ref{defn:framers}) is input-to-state stable (ISS)\footnote{If desired, we can replace the 2-norm in the original ISS definition with any norm that is more aligned with the norm used for the robustness, and the satisfaction of the former would also imply the latter by norm equivalence.}
\end{problem}
\section{Proposed Interval Observer} 
\label{sec:observer}
\subsection{Interval Observer Structure and Framer Property (Correctness)} \label{sec:obsv}
We begin by deriving an equivalent representation of the system dynamics for the system $\mathcal{G}$ in \eqref{eq:system}. 
Specifically, we first leverage Assumption \ref{ass:mixed_monotonicity} to decompose the functions $f$ and $h$ into two components based on the results in Proposition \ref{prop:JSS_decomp}:
\begin{align}\label{eq:JSS_decom}
\begin{array}{rl}
f(x)&=Ax+\phi(x), \\ 
h(x)&=Cx+\psi(x), 
\end{array}
\end{align}
where $A \in \mathbb{R}^{n \times n}$ and $C \in \mathbb{R}^{l \times n}$ are chosen such that $\phi$ and $\psi$ are JSS mappings (cf. Definition \ref{def:JSS}). Then, we propose an equivalent system representation for the plant $\mathcal{G}$ in \eqref{eq:system} with additional degrees of freedom (that will be our to-be-designed observer gains) 
through the following lemma.   
\begin{lemma}[Equivalent System Representation]
\label{lem:dynamics_reformulation}
Consider plant $\mathcal{G}$ in \eqref{eq:system} and suppose that Assumptions \ref{ass:known_input_output}--\ref{ass:mixed_monotonicity} hold.
Let $L,N \in \mathbb{R}^{n \times l}$ and $T \in \mathbb{R}^{n \times n}$ be arbitrary matrices that satisfy the following:
\begin{align}\label{eq:T_constraint}
T+NC=I_n.
\end{align}
Then, the system dynamics for the plant $\mathcal{G}$ in \eqref{eq:system} can be equivalently written as  
\begin{align}\label{eq:dynamics_reformulation}
\begin{cases}
\xi^+_{t}=M_x\xi_t+T\phi(x_t)-L\psi(x_t)-N\rho(x_t,w_t,u_t)\\
\hspace{0.75cm} +M_ww_t-M_vv_t+M_uu_t+(M_xN+L)y_t,\\
x_t = \xi_t +N y_t-NV v_t-NDu_t,
\end{cases}
\end{align}
with initial condition $\xi_0=x_0 -N y_0 +NDu_0 +NV v_0 $, where 
\begin{align}\label{eq:Ms}
\begin{array}{rl}
M_x&\triangleq T A-LC-NA_2, \ M_w \triangleq TW-NW_2, \\
M_u &\triangleq TB-NB_2-(M_xN+L)D, M_v \triangleq (M_xN+L)V.
\end{array}
\end{align}
Moreover, $A$ and $C$ are computed through the JSS decompositions in \eqref{eq:JSS_decom}, while $ A_2 \in \mathbb{R}^{l \times n}$, $W_2 \in \mathbb{R}^{l \times n_w}$, and $B_2 \in \mathbb{R}^{l \times m}$ are chosen such that the following decomposition hold {$\forall (x,w,u) \in \mathcal{X} \times \mathcal{W} \times \mathcal{U}$} (cf. Definition \ref{defn:JSS} and Proposition \ref{prop:JSS_decomp}):
\begin{align} \label{eq:JSS_decom_psi}
 \psi^+(x,w,u)&=A_2 x + W_2 w+B_2u +\rho(x,w,u), 
\end{align} 
such that $\rho$ is a JSS mapping, with   
\begin{itemize}
\item $\psi^+(x,w,u)\triangleq \dot{\psi}(x,w,u)=\frac{\partial \psi}{\partial x}(f(x)+Ww+Bu)$ if $\mathcal{G}$ is a CT system, and 
\item $\psi^+(x,w,u)\triangleq \psi(x^+,w,u)=\psi(f(x)+Ww+Bu)$ if $\mathcal{G}$ is a DT system.
\end{itemize}
\end{lemma}
\begin{proof} 
We begin by defining an auxiliary state 
\begin{align}\label{eq:aux_state_1}
\xi_t \triangleq x_t - N(y_t -V v_t-Du_t). 
\end{align}
Then, from the second equation in \eqref{eq:JSS_decom} and from \eqref{eq:aux_state_1}, as well as by choosing $N$ to satisfy \eqref{eq:T_constraint}, we have:
\begin{align}\label{eq:aux_state_2}
\begin{array}{rl}
\xi_t &= x_t - N(y_t -V v_t-Du_t)=x_t-N(Cx_t+\psi(x_t))=\underbrace{(I-NC)}_{T}x_t-N\psi(x_t)\\
&=Tx_t -N\psi(x_t).
\end{array}
\end{align}
Now, we are ready to derive the dynamics of the auxiliary state $\xi_t$ starting from \eqref{eq:aux_state_2}, and using \eqref{eq:system}, \eqref{eq:JSS_decom} and \eqref{eq:JSS_decom_psi}:
\begin{align}\label{eq:reform_int}
\begin{array}{rl}\xi^+_{t}&=Tx_t^+-N\psi^+(x_{t},w_t,u_t)\\
&=T(Ax_t+\phi(x_t)+Ww_t+Bu_t)-N(A_2 x_t+W_2 w_{t}+B_2u_t+\rho(x_{t},w_t,u_t)),
\end{array}
\end{align}
with $x_t^+=Ax_t+\phi(x_t)+Ww_t+Bu_t$ from \eqref{eq:system} and the second equation in \eqref{eq:JSS_decom}, as well as $\psi^+(x_{t},w_t,u_t)=A_2 x_t+W_2 w_{t}+B_2u_t+\rho(x_{t},w_t,u_t)$ from \eqref{eq:JSS_decom_psi}. Next, from the second equations in \eqref{eq:system} and \eqref{eq:JSS_decom}, we have 
\begin{align*}
L(y_t-Cx_t-\psi(x_t)-Vv_t-Du_u)=0, \ 
\end{align*}
for any matrix $L$ with appropriate dimensions. 
Adding this `zero term' to the right-hand side of \eqref{eq:reform_int} yields:
\begin{align}\label{eq:error_1}
\begin{array}{rl}
\xi^+_{t}&=M_xx_t+T\phi(x_t)-L\psi(x_t)-N\rho(x_t,w_t,u_t)+M_ww_t-LVv_t
\\&
+(TB-NB_2)u_t+Ly_t.
\end{array}
\end{align}
Finally, plugging $x_t$ from \eqref{eq:aux_state_1} into the first term on the right-hand side of \eqref{eq:error_1} results in the first equation in \eqref{eq:dynamics_reformulation}, while the second equation in \eqref{eq:dynamics_reformulation} is a reorganization of the first equality in \eqref{eq:aux_state_2}.
\end{proof}
\subsubsection{Construction of Interval Framers}
From the equivalent system in \eqref{eq:dynamics_reformulation}, we can construct an embedding system 
by formulating the corresponding decomposition functions for both CT and DT cases, respectively, as follows. 

First, note that by Proposition \ref{prop:lin_decomposition}, the component of the dynamics in \eqref{eq:dynamics_reformulation} that is affine in $\xi$, $w$, and $v$ (with known $y$ and $u$), i.e., 
\begin{align}\label{eq:linear_component}
g^{\ell}(\xi,w,v)\triangleq M_x\xi+M_ww-M_vv+M_uu+(M_xN+L)y,
\end{align}
admits a tight decomposition function (cf. Proposition \ref{prop:lin_decomposition} for details):
\begin{align} \label{eq:flow_Lin_dec}
\begin{array}{rl}
g^{\ell}_{d}(\xi_1,\xi_2,w_1,w_2,v_1,v_2)&=M_x^\uparrow \xi_1-M_x^\downarrow \xi_2+M_w^{\oplus}w_1-M_w^{\ominus}w_2+M_v^{\ominus}v_1-M_v^{\oplus}v_2\\
&\quad +M_uu+(M_xN+L)y,
\end{array}
\end{align} 
where
\begin{itemize}
\item $M_x^\uparrow \triangleq M_x^\oplus$ and $M_x^\downarrow \triangleq M_x^\ominus$ if \eqref{eq:mix_mon_lin} is a DT system, and 
\item $M_x^\uparrow \triangleq M_x^{\text{\emph{nd}},\oplus}+M_x^{\text{\emph{d}}}$ and $M_x^\downarrow \triangleq M_x^{\text{\emph{nd}},\ominus}$ if \eqref{eq:mix_mon_lin} is a CT system.
\end{itemize}
Moreover, the nonlinear component of the dynamics in \eqref{eq:dynamics_reformulation} consists of a matrix-weighted summation of JSS mappings $\phi$, $\psi$ and $\rho$:
\begin{align} \label{eq:nonlinear_component}
g^{\nu}_{d}(x,w,u)\triangleq T\phi(x)-L\psi(x)-N\rho(x,w,u),
\end{align} 
for which we construct a DT decomposition function using Proposition \ref{prop:tight_decomp} (to compute tight decomposition functions for each of the JSS mappings $\phi$, $\psi$, and $\rho$), and Proposition \ref{prop:lin_decomposition} (to compute a decomposition function for the linear combination/weighted sum of the JSS functions):
\begin{align} \label{eq:flow_Nonlin_dec}
\begin{array}{rl}
&g^{\nu}_{d}(x_1,x_2,w_1,w_2,u_1,u_2)=\\
&T^\oplus \phi_d(x_1,x_2)-T^\ominus \phi_d(x_2,x_1)+L^\ominus\psi_d(x_1,x_2)-L^\oplus\psi_d(x_2,x_1)\\
&+N^\ominus\rho_d(x_1,x_2,w_1,w_2,u_1,u_2)-N^\oplus\rho_d(x_2,x_1,w_2,w_1,u_2,u_1).
\end{array}
\end{align} 
Next, based on the fact that the summation of the decomposition functions construct a decomposition function for the summation of functions, we compute a decomposition function for the dynamics in \eqref{eq:dynamics_reformulation} by adding the two decomposition functions in \eqref{eq:flow_Lin_dec} and \eqref{eq:flow_Nonlin_dec}: 
\begin{align}
\begin{array}{l}
g_d(\xi_1,\xi_2,x_1,x_2,w_1,w_2,v_1,v_2,u_1,u_2)=\\
g^{\ell}_{d}(\xi_1,\xi_2,w_1,w_2,v_1,v_2)+g^{\nu}_{d}(x_1,x_2,w_1,w_2,u_1,u_2).
\end{array}
\end{align}
Finally, by leveraging $g_d$, we compute an embedding system for the dynamics in \eqref{eq:dynamics_reformulation} whose states frame the states of \eqref{eq:dynamics_reformulation}, and hence, equivalently the states of the original system \eqref{eq:system}, by \cite[Proposition 3]{khajenejad2021tight}. This process is summarized in the following Theorem.    
\begin{theorem}[{Correctness}]\label{thm:correctness}
Consider system \eqref{eq:system}, and let $L,T$ and $N$ be arbitrary matrices with appropriate dimensions that satisfy \eqref{eq:T_constraint}. Then, the following dynamical system initialized at $[\underline{x}^T_0 \ \overline{x}^T_0]^\top$, constructs a framer for \eqref{prop:tight_decomp}:
\begin{align}\label{eq:observer}
\begin{cases}
\underline{\xi}^+_{t}=M^\uparrow_x\underline{\xi}_t-M^\downarrow_x\overline{\xi}_t+M_w^{\oplus}\underline{w}-M_w^{\ominus}\overline{w}+M_v^{\ominus}\underline{v}-M_v^{\oplus}\overline{v}+M_uu_t+(M_xN+L)y_t\\
\hspace{0.75cm}+T^\oplus \phi_d(\underline{x}_t,\overline{x}_t)-T^\ominus \phi_d(\overline{x}_t,\underline{x}_t)+L^\ominus\psi_d(\underline{x}_t,\overline{x}_t)-L^\oplus\psi_d(\overline{x}_t,\underline{x}_t)\\
\hspace{0.75cm}+N^\ominus\rho_d(\underline{x}_t,\overline{x}_t,\underline{w},\overline{w},\underline{u},\overline{u})-N^\oplus\rho_d(\overline{x}_t,\underline{x}_t,\overline{w},\underline{w},\overline{u},\underline{u}),\\
\overline{\xi}^+_{t}=M^\uparrow_x\overline{\xi}_t-M^\downarrow_x\underline{\xi}_t+M_w^{\oplus}\overline{w}-M_w^{\ominus}\underline{w}+M_v^{\ominus}\overline{v}-M_v^{\oplus}\underline{v}+M_uu_t+(M_xN+L)y_t\\\hspace{0.75cm}+T^\oplus \phi_d(\overline{x}_t,\underline{x}_t)-T^\ominus \phi_d(\underline{x}_t,\overline{x}_t)+L^\ominus\psi_d(\overline{x}_t,\underline{x}_t)-L^\oplus\psi_d(\underline{x}_t,\overline{x}_t)\\
\hspace{0.75cm}+N^\ominus\rho_d(\overline{x}_t,\underline{x}_t,\overline{w},\underline{w},\overline{u},\underline{u})-N^\oplus\rho_d(\underline{x}_t,\overline{x}_t,\underline{w},\overline{w},\underline{u},\overline{u}),\\
\underline{x}_t = \underline{\xi}_t+(NV)^\ominus \underline{v}-(NV)^\oplus \overline{v}+N y_t-NDu_t,\\
\overline{x}_t = \overline{\xi}_t+(NV)^\ominus \overline{v}-(NV)^\oplus \underline{v}+N y_t-NDu_t,
\end{cases}
\end{align}
where $M_x,M_w,M_v$ and $M_u$ are defined in \eqref{eq:Ms}. Moreover, the initial conditions for $\underline{\xi}$ and $\overline{\xi}$ are $\underline{\xi}_0 = \underline{x}_0 - Ny_0 +ND u_0 +(NV)^\oplus \underline{v}- (NV)^\ominus \overline{v}$ and $\overline{\xi}_0 = \overline{x}_0 - Ny_0 +ND u_0 +(NV)^\oplus \overline{v}- (NV)^\ominus \underline{v}$.
\end{theorem}
\begin{proof}
We begin the proof by framing the auxiliary states of the equivalent system in \eqref{eq:dynamics_reformulation}, then applying Proposition \ref{prop:JSS_decomp} and \cite[Lemma 1]{efimov2013interval}, results in \eqref{eq:observer}. Hence, \eqref{eq:observer} creates a framer system for equivalent system in \eqref{eq:dynamics_reformulation} and in turn is a framer system for the original plant ${\mathcal{G}}$, which results from the fact that the summation of the constituent systems embedding framers/systems creates an embedding framers/systems for the constituent systems summation.
To compute the constituent systems framers of \eqref{eq:dynamics_reformulation}, we categorize them as follows: 
\begin{enumerate}[i)]
\item Known terms that are independent of noise and state, which can also be viewed as having their  lower and upper bounds equal to their original (known) value;

\item Linear terms  in state and noise that can be lower and upper framed by using \cite[Lemma 1]{efimov2013interval}, with a subtle difference in calculating the lower and upper embedding/framer systems for the DT and CT cases, which reflects in the definitions of $M^{\downarrow},M^{\uparrow}$ given below \eqref{eq:flow_Lin_dec};

\item Nonlinear terms with respect to the noise and state that can be lower and upper bounded/framed 
by applying Propositions \ref{prop:JSS_decomp} and \ref{prop:tight_decomp} as well as \cite[Lemma 1]{efimov2013interval}. 
\end{enumerate}

\noindent Then, adding up the constituent embedding framers/systems resulting from i)-iii) yields the embedding (observer) system \eqref{eq:observer}. {Finally, the embedding systems framer property \cite[Proposition 3]{khajenejad2021tight} implies the correctness property.}
\end{proof}
\begin{remark}[Additional Degrees of Freedom]\label{rem:1}
The observer structure described above draws inspiration from \cite{wang2012observer} and \cite{degue2021design}, aiming to introduce additional degrees of freedom. In contrast to the most existing designs, e.g., in \cite{khajenejad_H_inf_2022}, which relied on only one observer gain $L$, the current structure involves three to-be-designed observer gains $N,T$ and $L$. While this augmentation is advantageous for enhancing performance, it was found to be not an exact replacement for coordinate transformations that could potentially render the observer gain design problem in Theorem \ref{thm:stability} feasible (cf. Section \ref{sec:transformation}). In this case, a coordinate transformation can be applied in a standard manner (cf. Section \ref{sec:transformation} below,  \cite[Section V]{tahir2021synthesis} or \cite{mazenc2021when} and references therein for more details).
\end{remark} 
\subsection{$L_1$ or $\mathcal{H}_{\infty}$-Robust Observer Design}
In this section, we address the input-to-state stability of the proposed interval framer system \eqref{eq:observer}, for both DT and CT cases. To do this, we first compute the error dynamics of the observer, which is a nonlinear positive/cooperative system by construction. Then, by applying the results in Proposition \ref{prop:tight_decomp} and \eqref{eq:JSS_up_bound}, we compute a linear positive comparison system for the error dynamics. Finally, we provide necessary and sufficient conditions  for computing the observer gains that stabilize and minimize the $L_1$ or $\mathcal{H}_{\infty}$ norm of the comparison system.
\subsubsection{Error Dynamics}
We start by computing the observer error system and a corresponding linear comparison system through the following lemma.
\begin{lem}\label{lem:error_dynamics}
Consider system $\mathcal{G}$ in \eqref{eq:system} and the proposed framer system \eqref{eq:observer}. Then, the framer error ($\varepsilon_t \triangleq \overline{x}_t-\underline{x}_t$) dynamics, as well as a corresponding  linear comparison system can be computed as follows:  
\begin{align}\label{eq:error_sys}
\begin{array}{rl}
&e^+_t= |M_x|^*e_t+|M_w|\delta_w+(|M_v|+\sigma|NV|-|M_x|^*|NV|)\delta_v+|T|\delta^{\phi}_t+|L|\delta^{\psi}_t+|N|\delta^{\rho}_t, \\
&\leq \hspace{-.1cm} \begin{cases} |M_x|^*e_t+|M_w|\delta_w+(|NV|+|LV|)\delta_v+|T|\delta^{\phi}_t+|L|\delta^{\psi}_t+|N|\delta^{\rho}_t, \hfill  \text{(DT)}\\ |M_x|^*e_t+|M_w|\delta_w+((|M_x|\hspace{-.05cm}-\hspace{-.05cm}M_x^{\text{m}})|NV|\hspace{-.05cm}+\hspace{-.05cm}|LV|)\delta_v+|T|\delta^{\phi}_t +|L|\delta^{\psi}_t+ |N|\delta^{\rho}_t, \hfill \text{(CT)} \end{cases}\\
&\leq \tilde{A}_xe_t +\tilde{A}_w\delta_w+\tilde{A}_v\delta_v+\tilde{A}_u\delta_u=\tilde{A}_xe_t+\tilde{B}\delta_{\tilde w}, 
\end{array}
\end{align}
where $\delta_s \triangleq \overline{s}-\underline{s}, \ \forall s \in \{w,v,u\}$, $\delta^{\mu}_t \triangleq \mu(\overline{z}_t)-\mu(\underline{z}_t), \forall \mu \in \{\phi, \psi, \rho \}$ and $\delta_{\tilde w} \triangleq \begin{bmatrix} \delta_w & \delta_v & \delta{u} \end{bmatrix}^\top$, as well as 
\begin{align}\label{eq:param}
\begin{array}{rl}
\tilde{A}_w &\triangleq |M_w|+|N|F^w_{\rho}, \quad \quad \quad \ \ \ \tilde{A}_x \triangleq |M_x|^*+|T|F_{\phi}+|L|F_{\psi}+|N|F^x_{\rho},\\
\tilde{A}_u &\triangleq |N|F^u_{\rho}, \quad \quad \quad \quad \quad \quad \quad \tilde{A}_v \triangleq |LV|+(\sigma I+(1-\sigma)(|M_x|-M_x^{\emph{\text{m}}}))|NV|,\\
\tilde{B} &\triangleq \begin{bmatrix} \tilde{A}_w & \tilde{A}_v & \tilde{A}_u \end{bmatrix},\\
\sigma &= \begin{cases}1 \quad \text{if} \quad \mathcal{G} \ \text{is DT}, \\ 
0 \quad \text{if} \quad \mathcal{G} \ \text{is CT},\end{cases} \ \ |M_x|^* \triangleq \begin{cases} |M_x| \triangleq M_x^\oplus+M_x^\ominus \quad \ \text{if} \ \mathcal{G} \ \text{is DT},\\ 
M_x^{\emph{\text{m}}}\triangleq|M_x^{\emph{\text{nd}}}|+M_x^{\emph{\text{d}}} \quad \text{if} \ \mathcal{G} \ \ \text{is CT}.
\end{cases}
\end{array}
\end{align}
Moreover,  as in \eqref{eq:Ms},  
\begin{align*}
\begin{array}{rl}
M_x&\triangleq T A-LC-NA_2, \ M_w \triangleq TW-NW_2, \\
M_u &\triangleq TB-NB_2-(M_xN+L)D, M_v \triangleq (M_xN+L)V.
\end{array}
\end{align*}
Furthermore, $F_{\phi},F_{\psi}$ and $F^s_{\rho}, \forall s \in \{w,w,u\}$, \mk{for the JSS mappings $\phi, \psi$ and $\rho$, respectively, are as given below \eqref{eq:JSS_up_bound} (cf. Proposition \ref{prop:tight_decomp}).} 
\end{lem}
\begin{proof}
Defining $\varepsilon_t \triangleq \overline{\xi}_t-\underline{\xi}_t$, note that the first two equations in \eqref{eq:observer} imply that:
\begin{align}\label{eq:xix}
\varepsilon^+_t=|M_x|^*\varepsilon_t+|M_w|\delta_w+|M_v|\delta_v+|T|\delta^{\phi}_t+|L|\delta^{\psi}_t+|N|\delta^{\rho}_t,
\end{align}
where $|M_x|^*,M_w$ and $M_v$ are given in \eqref{eq:param}.
On the other hand, from the third and fourth equations in \eqref{eq:observer} we obtain the following: 
\begin{align}\label{eq:xix_2}
\begin{array}{rl}
&e_t \triangleq \overline{x}_t-\underline{x}_t=\overline{\xi}_t-\underline{\xi}_t+|NV|\delta_v=\varepsilon_t+|NV|\delta_v \\
\implies & e^+_t\triangleq \begin{cases} e_{t+1}=\varepsilon_{t+1}+|NV|\delta_v=\varepsilon^+_{t}+|NV|\delta_v, & \text{if} \ \mathcal{G} \ \text{is DT},\\
\dot{e}_{t}=\dot{\varepsilon}_{t}=\varepsilon^+_{t}, & \text{if} \ \mathcal{G} \ \text{is CT},
\end{cases}\\& \quad\hspace{0.1cm} =\varepsilon^+_t+\sigma |NV|\delta_v.
\end{array}
\end{align}
Combining \eqref{eq:xix} and \eqref{eq:xix_2} returns the error dynamics in \eqref{eq:error_sys}. Then, the first inequality holds by the facts that in the DT case, we have $\sigma=1$ and 
\begin{align}\label{eq:DT_upper_bound}
\begin{array}{rl}
&(|M_v|+|NV|-|M_x|^*|NV|)\delta_v\\ &=(|M_xNV+LV|-|M_x||NV|+|NV|)\delta_v \\
&\leq (|M_xNV|+|LV|-|M_x||NV|+|NV|)\delta_v =(|LV|+|NV|)\delta_v, 
\end{array}
\end{align}
and in the CT case, we have $\sigma=0$ and
\begin{align}\label{eq:CT_upper_bound}
\begin{array}{rl}
&(|M_v|-|M_x|^*|NV|)\delta_v \\
&=(|M_xNV+LV|-M^{\emph{\text{m}}}_x|NV|)\delta_v \\
&\leq (|M_xNV|+|LV|-M^{\emph{\text{m}}}_x|NV|)\delta_v\leq ((|M_x|-M^{\emph{\text{m}}}_x)|NV|+|LV|)\delta_v, 
\end{array}
\end{align}
where the inequalities in \eqref{eq:DT_upper_bound} and \eqref{eq:CT_upper_bound} hold by the sub-multiplicative property of the $|\cdot|$ operator and the positivity of all the multiplicands. Finally, bounding the nonlinear JSS function differences $\delta^s_t \triangleq \overline{s}-\underline{s}, \ \forall s \in \{\phi,\psi,\rho\}$ based on Proposition \ref{prop:tight_decomp} and \eqref{eq:JSS_up_bound} yields the second inequality, i.e., the comparison system in \eqref{eq:error_sys}.
\end{proof}
\begin{remark}
The main purpose of the upper-bounding in \eqref{eq:DT_upper_bound} (DT case) is to remove the bilinearities at the expense of potentially adding some conservatism. Moreover, the two upper-boundings in \eqref{eq:CT_upper_bound} (CT case) were done to reduce the number of bilinear terms to only one single bilinear term, which we found to be unavoidable.
\end{remark}
\subsubsection{Observer Design}
In this subsection, we provide mixed integer linear programs (MILPs) and mixed-integer semi-definite programs (MISDPs)  for computing stabilizing observer gains, while minimizing the $L_1$ and $\mathcal{H}_{\infty}$ norms of the comparison system in \eqref{eq:error_sys}, respectively. The provided conditions are tight, in the sense that the results are necessary and sufficient for the stability of the comparison system and its $L_1$ or $\mathcal{H}_{\infty}$-robustness, for both CT and DT cases, as summarized in the following theorem.   
\begin{theorem}[$L_1$ or $\mathcal{H}_{\infty}$-Robust and 
ISS Observer Design]\label{thm:stability}
Suppose Assumptions \ref{ass:known_input_output} and \ref{ass:mixed_monotonicity} hold for the nonlinear system $\mathcal{G}$ in \eqref{eq:system}. Then, the following statements hold.
\begin{enumerate}[(i)]
\item \label{item:L1_min} The correct interval framer proposed in \eqref{eq:observer} is $L_1$-robust if there exists a tuple $(\gamma_*, \Delta_*,Q_*,\Omega_*,\Gamma_*,\tilde{L}_*,\tilde{N}_*,\tilde{T}_*,\tilde{M}_{x*},\tilde{Z}_*,N_*,\Phi_*)$ that solves the following mixed-integer program (MIP):
\begin{align}\label{eq:MILP}
\begin{array}{rl}
&\min\limits_{\{\gamma, \Delta,Q,\Omega,\Gamma,\tilde{L},\tilde{N},\tilde{T},\tilde{M}_x,\tilde{Z},N,\Phi\}} \gamma \\
&s.t. \ \mathbf{1}^\top_{n} \begin{bmatrix} \Delta & \Gamma & \Phi & \Omega \end{bmatrix} < \begin{bmatrix}  & \gamma \mathbf{1}^\top_{{n}_w} & \gamma \mathbf{1}^\top_{m} & \gamma \mathbf{1}^\top_{{n}_v} & \mathbf{1}^\top_{n}(\sigma Q-I) \end{bmatrix}, \ \text{and} \ \{\mathbf{C}\} \ \text{holds}.
\end{array}
\end{align}
\item \label{item:H_min} The correct interval framer proposed in \eqref{eq:observer} is $\mathcal{H}_{\infty}$-robust if there exists a tuple $(\gamma_*, \Delta_*,Q_*,\Omega_*,\Gamma_*,\tilde{L}_*,\tilde{N}_*,\tilde{T}_*,\tilde{M}_{x*},\tilde{Z}_*,N_*,\Phi_*)$ that solves the following mixed-integer program (MIP):
\begin{align}\label{eq:MISDP}
\begin{array}{rl}
&\min\limits_{\{\gamma, \Delta,Q,\Omega,\Gamma,\tilde{L},\tilde{N},\tilde{T},\tilde{M}_x,\tilde{Z},N,\Phi\}} \gamma \\
&s.t. \ \begin{cases} \begin{bmatrix} Q & \Omega & \begin{bmatrix} \Delta  & \Phi & \Gamma \end{bmatrix}  & 0 \\
                                    * & Q & 0 & I \\
                                    * & * & \gamma I & 0 \\
                                    * & * & * & \gamma I \end{bmatrix} \succ 0, \ \text{if} \ \sigma=1,\\
                           \begin{bmatrix} \Omega + \Omega^\top & \begin{bmatrix} \Delta  & \Phi & \Gamma \end{bmatrix}  & I \\
                                    * &  -\gamma I & 0 \\
                                    * &  * & -\gamma I \end{bmatrix} \prec 0, \ \text{if} \ \sigma=0,         \end{cases}, \ \text{and} \ \{\mathbf{C}\} \ \text{holds}.
\end{array}
\end{align}
\end{enumerate}
where $\sigma = \begin{cases}1 \quad \text{if} \quad \mathcal{G} \ \text{is DT}, \\ 
0 \quad \text{if} \quad \mathcal{G} \ \text{is CT},\end{cases}$, $\beta = \begin{cases}1 \quad \text{if} \quad V\ne0 , \\ 
0 \quad \text{if} \quad V=0\end{cases}$ and
\begin{align}\label{eq:conditions}
\mathbf{C}=\begin{cases}
 \Delta=|\tilde{T}W-\tilde{N}W_2|+|\tilde N|F^w_{\rho},\\
\Gamma = |\tilde N|F^u_{\rho}, \\
 \Phi = |\tilde{L}V|+\sigma |\tilde{N}V|+(1-\sigma)Z,\\
 \Omega = \sigma|\tilde{M}_x|+(1-\sigma)\tilde{M}_x^m+|\tilde T|F_{\phi}+|\tilde L|F_{\psi}+|\tilde N|F^x_{\rho},\\
 Z=(|\tilde{M}_x|-\tilde{M}_x^{\emph{\text{m}}})NV,\\
 \tilde{N}=\beta(1-\sigma)QN+\beta\sigma \tilde{N}+(1-\beta)\tilde{N}, \\
 \tilde{M}_x=\tilde{T} A-\tilde{L}C-\tilde{N}A_2,\\
 \tilde{T}=Q-\tilde{N}C,\\
 \gamma > 0, Q \in \mathbb{D}^n_{>0}.
\end{cases}
.\end{align}
Moreover, in both DT and CT cases, the $L_1$ 
or $\mathcal{H}_{\infty}$-robust observer gains can be computed as $X=Q^{-1}_*\tilde{X}_*, \forall X \in \{L,T,N\}$, where the tuple $(Q_*,L_*,T_*,N_*)$ is the optimal solution to the MIP in \eqref{eq:MILP} or \eqref{eq:MISDP}, respectively.  
\end{theorem}
\begin{proof}
The comparison system in \eqref{eq:error_sys} can be rewritten as:
\begin{align}\label{eq:comparison}
\tilde{\mathcal{G}}:\ x^+_t=\tilde{A}_xx_t+\tilde{B}\delta_{\tilde w}, \quad z_t=\tilde{C}x_t +\tilde{D}\delta_{\tilde w}, \tilde{C}=I, \tilde{D}=0.
\end{align}
Note that by construction, \eqref{eq:comparison} is a positive/cooperative system since $\tilde{A}_w,\tilde{A}_v$ and $\tilde{A}_u$ are non-negative, whereas $\tilde{A}_x$ is non-negative in the DT case and is Metzler in the CT case. This follows from the non-negativity of the operator $|\cdot|$ and the matrices $F_{\phi},F_{\psi},F^x_{\rho},|\tilde{M}_x|-\tilde{M}_x^{\text{m}}$, as well as the facts that in the DT case, $|\tilde{M}_x|$ is a non-negative matrix, while in the CT case, $\tilde{M}^{\text{m}}_x$ is a Metzler matrix by construction, in addition to the fact the summation of non-negative and Metzler matrices is Metzler. Now, given the positive/cooperative system \eqref{eq:comparison}, below we show the sufficiency and necessity of the conditions in \eqref{eq:MILP} and \eqref{eq:MISDP} for the stability and $L_1$ or $\mathcal{H}_{\infty}$ robustness of \eqref{eq:comparison}. \\[-0.25cm]

\noindent \textbf{Proof of \eqref{item:L1_min}: $L_1$-Norm Minimization}:
In the DT case, by \cite[Theorem 2]{chen2013l1}, the positive linear system $\tilde{\mathcal{G}}$ in \eqref{eq:comparison} is stable and satisfies $\|z_t\|_{\ell_1} < \gamma \|\delta_{\tilde w}\|_{\ell_1}$ for some $\gamma >0$, if and only if there exist $p \in \mathbb{R}^n_{>0}$ such that 
\begin{align}\label{eq:MILP_DT_condition}
\begin{bmatrix} \tilde{A}_x & \tilde{B} \\ I & 0
\end{bmatrix}^\top \begin{bmatrix} p \\ \mathbf{1}_{n} \end{bmatrix}<\begin{bmatrix} p \\ \gamma\mathbf{1}_{\tilde{n}} \end{bmatrix} \Leftrightarrow \begin{bmatrix} p^\top & \mathbf{1}^\top_n \end{bmatrix}\begin{bmatrix} \tilde{A}_x & \tilde{B} \\ I & 0
\end{bmatrix} < \begin{bmatrix} p^\top & \gamma\mathbf{1}^\top_{\tilde{n}} \end{bmatrix} \Leftrightarrow \begin{cases} p^\top \tilde{A}_x+\mathbf{1}^\top_n < p^\top,\\
p^\top\tilde{B} < \gamma \mathbf{1}^\top_{\tilde n}, \end{cases}
\end{align}
while in the CT case, it follows from \cite[Lemma 1]{briat2013robust} that \eqref{eq:comparison} is stable and satisfies $\|z_t\|_{\ell_1} < \gamma \|\delta_{\tilde w}\|_{\ell_1}$ for some $\gamma >0$, if and only if there exist $p \in \mathbb{R}^n_{>0}$ such that 
\begin{align}\label{eq:MILP_CT_condition}
\begin{bmatrix} \tilde{A}_x & \tilde{B} \\ I & 0
\end{bmatrix}^\top \begin{bmatrix} p \\ \mathbf{1}_{n} \end{bmatrix}<\begin{bmatrix} \mathbf{0}_n \\ \gamma\mathbf{1}_{\tilde{n}} \end{bmatrix} \Leftrightarrow \begin{bmatrix} p^\top & \mathbf{1}^\top_n \end{bmatrix}\begin{bmatrix} \tilde{A}_x & \tilde{B} \\ I & 0
\end{bmatrix} < \begin{bmatrix} \mathbf{0}_n^\top & \gamma\mathbf{1}^\top_{\tilde{n}} \end{bmatrix} \Leftrightarrow \begin{cases} p^\top \tilde{A}_x+\mathbf{1}^\top_n < \mathbf{0}_n^\top,\\
p^\top\tilde{B} < \gamma \mathbf{1}^\top_{\tilde n}, \end{cases}
\end{align} 
where $\tilde{n} \triangleq n_w+n_v+m$. Then, by defining the positive diagonal matrix variable  $Q=Q^\top=\emph{\text{diag}}(p) \Leftrightarrow p=Q\mathbf{1}_n \Leftrightarrow p^\top=\mathbf{1}^\top_nQ$, the inequalities in \eqref{eq:MILP_DT_condition} are equivalent to: 
\begin{align}\label{eq:CTDT_stable}
\text{DT case:} \ \begin{cases}  \mathbf{1}^\top_nQ\tilde{A}_x < \mathbf{1}^\top_n(Q-I),\\
\mathbf{1}^\top_nQ\tilde{B} < \gamma \mathbf{1}^\top_{\tilde n}, \end{cases} \quad \text{CT case:} \ \begin{cases}  \mathbf{1}^\top_nQ\tilde{A}_x < -\mathbf{1}^\top_n,\\
\mathbf{1}^\top_nQ\tilde{B} < \gamma \mathbf{1}^\top_{\tilde n}. \end{cases}
\end{align}
The conditions in \eqref{eq:CTDT_stable} are equivalent to the ones in \eqref{eq:MILP} in both DT ($\sigma=1$) and CT ($\sigma=0$) cases after plugging $\tilde{A}_x$ and $\tilde{B} \triangleq \begin{bmatrix} \tilde{A}_w & \tilde{A}_v & \tilde{A}_u \end{bmatrix}$ from \eqref{eq:param} into \eqref{eq:CTDT_stable}, in addition to considering the constraint $T=I-NC \Leftrightarrow QT=Q-QNC$, defining $\tilde{X} \triangleq QX, \forall X \in \{L,T,N,M_x\}$, and given the fact that $Q|Y|=|QY|$ since $Q$ is a positive diagonal matrix.\\[-0.25cm]

\noindent \textbf{Proof of \eqref{item:H_min}: $\mathcal{H}_{\infty}$-Norm Minimization}: In the DT case, by \cite[Theorem 3]{najson2013kalman} as well as applying Schur complement, we conclude that the positive comparison system \eqref{eq:comparison} is stable, with its $\mathcal{H}_{\infty}$ norm being less than $\gamma$, if and only if there exists a positive diagonal matrix $Q$ such that 
\begin{align}\label{eq:MISDP_DT_condition}
\begin{bmatrix} Q & Q\tilde{A}_x & Q\tilde{B}  & 0 \\
                                    * & Q & 0 & I \\
                                    * & * & \gamma I & 0 \\
                                    * & * & * & \gamma I \end{bmatrix} \succ 0.
\end{align}
Similarly, in the CT case, \cite[Theorem 2]{tanaka2011bounded} as well as its Schur complement lead to the following necessary and sufficient conditions:
\begin{align}\label{eq:MISDP_DT_condition}
\begin{bmatrix}  Q\tilde{A}_x+\tilde{A}^\top_xQ & Q\tilde{B}  & I \\
                                    * &  -\gamma I & 0 \\
                                    * & *  & -\gamma I \end{bmatrix} \prec 0.
\end{align}
By similar arguments as in the proof of \eqref{item:L1_min}, the conditions in \eqref{eq:MISDP_DT_condition} and \eqref{eq:MISDP_DT_condition} as well as the additional constraint $T=I-NC$ can be equivalently represented in the form of \eqref{eq:MISDP}.  
\end{proof}

Note that in the CT case and only if $V \ne 0$, the presence of the terms $(|\tilde{M}_x|-\tilde{M}^{\text{m}}_x)NV$ and $QN$ in the fifth and sixth equalities in $\mathbf{C}$ leads to bilinear constraints, but the MIP
still remains solvable with off-the-shelf solvers, e.g., Gurobi \cite{gurobi}. Nonetheless, for many classes of systems, e.g., for DT systems, as well as when $V=0$ (for both DT and CT systems), as shown in the following proposition, or by choosing $N = 0$ (at the cost of losing the extra degrees
of freedom with $T$ and $N$) the MIP in \eqref{eq:MISDP} reduces to a mixed-integer linear program (MILP) and the one in \eqref{eq:MISDP} to a mixed-integer semi-definite program (MISDP).
\begin{proposition}[MILP \& MISDP for DT Systems or when $V=0$]
If system $\mathcal{G}$ is DT, or if $V=0$, then \eqref{eq:MILP} is an MILP and \eqref{eq:MISDP} is an MISDP.
\end{proposition}
\begin{proof}
If system $\mathcal{G}$ is DT, then $\sigma=1$. This implies that the sixth inequality in $\mathbf{C}$ becomes 
$\tilde{N}=\tilde{N}$ while $\Phi$ no longer contains $Z$, i.e., there will be no bilinearity anymore in \eqref{eq:MILP} and \eqref{eq:MISDP}. Similarly, if $V=0$ (for both DT and CT cases) and hence $\beta=0$, then the fifth and sixth equalities in $\mathbf{C}$ will be $Z=0$ 
and $\tilde{N}=\tilde{N}$, and again no bilinearity remains. 
\end{proof}

\begin{remark}[MILP/MISDP versus LP/SDP]\label{rem:lp=sdp}
The mixed-integer nature of the proposed programs in Theorem \ref{thm:stability} stems from the existence of the terms involving absolute values $|M|$, as well as ``Metzlerization" $M^m=M^d+|M^{nd}|$. If desired, extra positivity constraints can be imposed to remove the need for the absolute values in $|M|$ and $|M^{nd}|$, similar to the literature on SDP/LMI-based interval observer designs, which would lead to a semi-definite program (SDP) or a linear program (LP). This addition is found to sometimes not incur any conservatism, but the problem becomes infeasible in others.
\end{remark}
\mk{\begin{remark}[MILP/MISDP versus LP/SDP] Note that the difference between the two proposed observer designs in Theorem \ref{thm:stability} is due to the system gain/norm (of the transfer function) of the observer error system that we choose to minimize. As a result of this choice (as desired/determined by the user), the $L_1$ design requires solving a mixed-integer linear program, while the $\mathcal{H}_{\infty}$ one is based on solving a mixed-integer semi-definite program, and nevertheless, both of the designs are applicable to bounded-Jacobian DT and CT nonlinear systems. Finally, in general it is hard to compare them in the sense of tightness of the interval estimates, since they are minimizing different system norms.\end{remark}
\begin{remark}[Feasibility and Conservatism of our Designs] It is worth emphasizing that a necessary condition for finding the observer gains using Theorem \ref{thm:stability} is that the closed-loop framer error system gain $\tilde{A}_x$ in \eqref{eq:error_sys} must be stable and we do not require $\tilde{A}_x$ to be Metzler (for CT systems) or positive (for DT systems). Instead, one can observe that $\tilde{A}_x$ in \eqref{eq:param} is already Metzler or positive 
by design. It is well known that simultaneously imposing Metzler/positivity and stability constraints in interval observer designs can be very conservative. To illustrate this, consider a discrete-time framer error system, $x^+=Ax$ with $A=-0.1$. It can be observed that $A$ is not simultaneously Schur stable and positive, but analog to our absolute value operator within $\tilde{A}_x$ in \eqref{eq:param}, $|A|=0.1$ is Schur stable.
\end{remark}}
\subsection{Coordinate Transformation}\label{sec:transformation}
Note that \mk{in} the construction of the proposed interval observer in \eqref{eq:observer} inspired by \cite{wang2012observer}, we introduced a linear transformation using gains $T$ and $N$ such that $T+NC=I$. This transformation provides additional degree of freedom, i.e., with to-be-designed gains $L,N,T$, in contrast to only $L$. While these additional observer gains resulting in a linear transformation can act as a surrogate to coordinate transformations and are helpful from the performance perspective, we could still benefit from having a coordinate transformation (which is a form of nonlinear transformation) that is commonly required for previous interval observer designs. \mk{Note that the coordinate transformation should be applied to the original system \eqref{eq:system} and ``before" defining auxiliary states, not to the equivalent system \eqref{eq:dynamics_reformulation} that includes auxiliary states.} 
In such cases, by introducing a coordinate transformation $z=Sx$ for invertible $S\in\mathbb{R}^{n\times n}$, the system \eqref{eq:system} can be rewritten in the new coordinates as:
\begin{align} \label{eq:system_new}
\begin{array}{ll}
\mathcal{G}: \begin{cases} {z}_t^+ =\hat{f}(z_t)+\hat{W}w_t+\hat{B}u_t ,   \\
\  y_t = \hat{h}(z_t)+Vv_t+Du _t, 
\end{cases}, \ \text{for all} \ t \in {\mathbb{T}},  
\end{array}
\end{align}
where $\hat{f}(z_t)=Sf(S^{-1}z_t)$, $\hat{W}=SW$, $\hat{B}=SB$ and $\hat{h}(z_t)=h(S^{-1}z_t)$.

Then, for an appropriate choice of an invertible $S$ and using the transformed system \eqref{eq:system_new}, the interval observer approach in Theorem \ref{thm:correctness} and Theorem \ref{thm:stability} can \mk{be} used to compute feasible observer gains $T,N,L$. Such an invertible $S$ can be found for example, by selecting any $L$ with $N=0$ (consequently $T=I$) for $A$ in \eqref{eq:JSS_decom}, such that $A-LC$ is Schur/Hurwitz with real eigenvalues, and then choosing the invertible $S$ that diagonalizes $A-LC$ i.e., such that $S(A-LC)S^{-1}$ is diagonal; this method for finding an invertible $S$ is a modification of an approach discussed in \cite{tahir2021synthesis}.

Finally, to motivate the potential need for coordinate transformation on top of the additional degrees of freedom, consider an autonomous linear system $\dot{x}=Ax$, with observations $y=Cx$, with $$A=\begin{bmatrix}
    -0.198 & 0.287 & -1.087 \\
    -0.09 &    0.231 & 0.915 \\
    -0.17 & 0.274 & -1.139
\end{bmatrix}, \qquad C=\begin{bmatrix}
    0.4 & -0.2 & -0.2
\end{bmatrix}. $$
For this system, we observe  that the optimization problems in \eqref{eq:MILP} and \eqref{eq:MISDP} are infeasible without any coordinate transformations. In contrast, feasible and optimal interval observers can be designed using the proposed approaches in previous section by transforming the system into new coordinates using the following invertible time-invariant transformation matrix $$S=\begin{bmatrix}
    3704.95 & -184.77 & -3768.21 \\
    -7216.86 &    282.19 & 7252.94 \\
    -3511.73 & 96.81 & 3485.65
\end{bmatrix}. $$
Hence, coordinate transformation prior to designing the interval observer can still be useful in making the design feasible for our proposed approaches, \tpa{which will be further demonstrated in simulations in the next section.} 

\section{Illustrative Examples}
{In this section, we consider both CT and DT examples to demonstrate the effectiveness of our approaches. The MILPs in \eqref{eq:MILP} and MISDPs in \eqref{eq:MISDP} are solved using YALMIP \cite{YALMIP} and Gurobi \cite{gurobi}.}

\vspace{-0.1cm}
\subsection{CT System Example 1}\label{sec:CT_exm}
{Consider} 
{the CT system}
from \mk{\cite[Eq. (30)]{dinh2014interval}}:
\begin{gather*}
\dot{x}_{1} = x_{2}+w_1, \quad \dot{x}_{2}=b_1x_3-a_1\sin(x_1)-a_2x_2+w_2, \\ 
 \dot{x}_3=-a_3(a_2x_1+x_2)+\frac{a_1}{b_1}(a_4\sin(x_1)+\cos(x_1)x_2)- a_4x_3 + w_3,
\end{gather*}
with state $x=[x_1,x_2,x_3]^\top$, output $y=x_1$ and corresponding parameters: $b_1=15, a_1=35.63, a_2=0.25,a_3=36,a_4=200$, $\mathcal{W}=[-0.1,0.1]^3$ and $\mathcal{X}_0 = [19.5, 9] \times [9, 11] \times [0.5,1.5]$. 

\tpa{For the above CT system, any nonlinear functions of only state $x_1$ can be considered as measured outputs since we are measuring $x_1$. Consequently, the remaining nonlinear component of the dynamics is $f(x)=[0,0,\frac{a_1}{b_1}cos(x_1)x_2]^\top$, and its corresponding Jacobian sign-stable mapping $\phi(x)=f(x)-Hx$, where $H=\begin{bmatrix}
   0  &0   &    0\\
    0  &0   &    0\\
    29.692& 2.375&0 
\end{bmatrix}$, is computed as per Proposition \ref{prop:JSS_decomp}. Therefore,  $\phi_1(x)=\phi_2(x)=0$ and for $\phi_3(x)$, as per Proposition \ref{prop:tight_decomp}, we can construct a decomposition function $\phi_{d,3}(x_1,x_2)=\phi_3(D_3x_1+(I-D_3)x_2)$ with $D_3=\begin{bmatrix}
   0  &0   &    0\\
    0  &0   &    0\\
    0& 0&0 
\end{bmatrix}$, as well as the corresponding $F_\phi=\begin{bmatrix}
   0  &0   &    0\\
    0  &0   &    0\\
    59.384& 4.751&0 
\end{bmatrix}$ as shown below \eqref{eq:JSS_up_bound}. 

The resultant $L_1$-robust observer gains computed by solving the optimization in Theorem \ref{thm:stability} are given below:
\begin{gather*}
T=\begin{bmatrix}
    0       &    0   &  0 \\
    -20.506 &    1   &  0 \\
    0       &    0   &  1
\end{bmatrix}, \notag \
    \,L= \begin{bmatrix} 159.384 \\ 102.531 \\ 29.692 \end{bmatrix}, \notag \
    \,N= \begin{bmatrix} 1 \\ 20.506 \\ 0 \end{bmatrix}
\end{gather*}
and using \eqref{eq:Ms}, the corresponding observer matrices $M_x$ and $M_w$ are : \begin{gather*}
M_x=\begin{bmatrix}
    -159.384       &    0   &  0 \\
    0 &    -20.756   &  15 \\
    0       &    -33.624   &  -200
\end{bmatrix}, \notag \
    \,M_w=\begin{bmatrix}
   0       &    0   &  0 \\
    -20.506 &    1   &  0 \\
    0       &    0   &  1
\end{bmatrix}.
\end{gather*}
Similarly, we can compute the $\mathcal{H}_\infty$-robust observer gains as
\begin{gather*}
T=\begin{bmatrix}
    0       &    0   &  0 \\
    -104.538 &    1   &  0 \\
    0       &    0   &  1
\end{bmatrix}, \notag \
    \,L= \begin{bmatrix} 5015607.653 \\ 522.690 \\ 29.692 \end{bmatrix}, \notag \
    \,N= \begin{bmatrix} 1 \\ 104.538 \\ 0 \end{bmatrix}
\end{gather*}
with the observer matrices \begin{gather*}
M_x=\begin{bmatrix}
    -5015607.653       &    0   &  0 \\
    0 &    -104.788   &  15 \\
    0       &    -33.625   &  -200
\end{bmatrix}, \notag \
    \,M_w=\begin{bmatrix}
   0       &    0   &  0 \\
    -104.538 &    1   &  0 \\
    0       &    0   &  1
\end{bmatrix}.
\end{gather*}
In both $L_1$ and $\mathcal{H}_\infty$ cases, due to the absence of a control input and measurement noise, both $M_u$ and $M_v$ are matrices of zeros.
}
 \begin{figure}[t!] 
\centering
{\includegraphics[width=0.48\columnwidth]{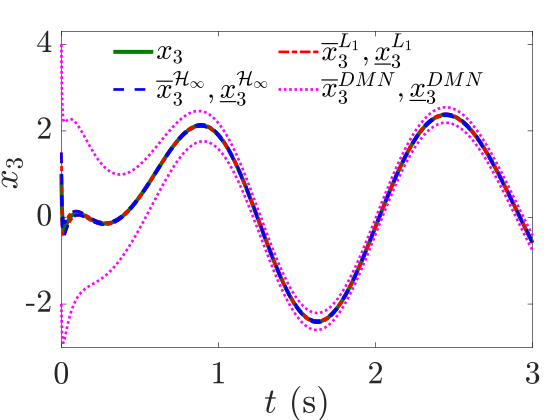}} \label{fig:sub1} 
{\includegraphics[width=0.48\columnwidth]{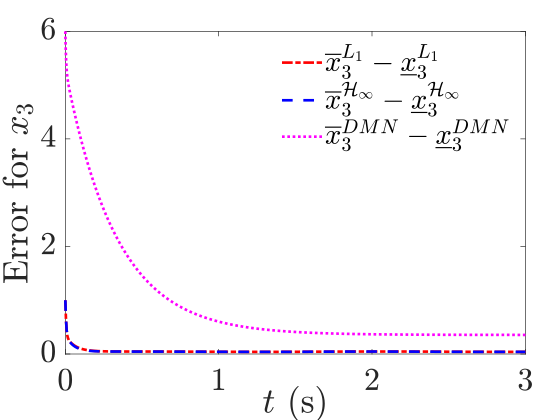}} \label{fig:sub2}
\caption{\small CT Example \tpa{1}: State, $x_3$, and its upper and lower framers and error \mk{of our proposed observer}, $\overline{x}^{L_1}_3,\underline{x}^{L_1}_3,\varepsilon^{L_1}_3$ for the $L_1$-robust interval observer, \mk{ $\overline{x}^{\mathcal{H}^{\infty}}_3,\underline{x}^{\mathcal{H}^{\infty}}_3,\varepsilon^{\mathcal{H}^{\infty}}_3$ for the $\mathcal{H}^{\infty}$}-robust interval observer, and $\overline{x}^{DMN}_3,\underline{x}^{DMN}_3,\varepsilon^{DMN}_3$ for the observer in \cite{dinh2014interval}.}
\label{fig:figure1}
\end{figure}

From Figure \ref{fig:figure1} (only results for state $x_3$ is shown for brevity; other states follow the same trends), both of the proposed  $L_1$- and $\mathcal{H}^\infty$-robust observers have comparable performance in framing the true state and in the convergence of the observer error $\varepsilon_t= \overline{x}_t-\underline{x}_t$, and notably, both approaches provide tighter bounds than  the interval observer presented in \cite{dinh2014interval},  $\underline{x}^{DMN},\overline{x}^{DMN}$. Further, from Figure  \ref{fig:figure1} (right),  the framer error in both the proposed cases exponentially converge to a steady-state faster than the convergence of the framer error in \cite{dinh2014interval}.

\vspace{-0.1cm}
\subsection{CT System Example 2}\label{sec:CT_exm2}
\tpa{Consider 
the linear CT system
from Section \ref{sec:transformation} with additive process noise, i.e., $\dot{x}=Ax + Ww$, and observations $y=Cx$, with system matrices $A$, $C$ and transformation matrix $S$ as given in Section \ref{sec:transformation} and with $W=S^{-1}$. This is an unstable linear system and as stated in Section \ref{sec:transformation}, the corresponding optimization problems in \eqref{eq:MILP} and \eqref{eq:MISDP} are infeasible without any coordinate transformations. However, with the transformation matrix $S$, we obtain feasible $L_1$-robust observer gains as given below:
\begin{gather*}
T=\begin{bmatrix}
     1   &   0   &    0\\
 -0.143 &  0.857 &  0.144 \\
 0   &   0   &    1
\end{bmatrix}, \notag \
    \,L= \begin{bmatrix}  86.988\\64.847\\
        147.363
 \end{bmatrix}, \notag \
    \,N= \begin{bmatrix} 0 \\ -0.6271 \\ 0 \end{bmatrix}
\end{gather*}
and the corresponding observer matrices are: \begin{gather*}
M_x=\begin{bmatrix}
     -0.03       &    0   &  0 \\
    0.0029 &    -0.0186   &  0 \\
    0.008       &    0.008   &  -0.018
\end{bmatrix}, \notag \
    \,M_w=\begin{bmatrix}
   1       &    0   &  0 \\
    -0.143&  0.857&    0.144 \\
    0       &    0   &  1
\end{bmatrix}.
\end{gather*}
Similarly, the $\mathcal{H}_\infty$-robust observer gains and matrices are
\begin{gather*}
T=\begin{bmatrix}
    1       &    0   &  0 \\
    0 &    1   &  0 \\
    0       &    0   &  1
\end{bmatrix}, \notag \
    \,L= \begin{bmatrix}  86.988\\
            65.509\\
          147.386
 \end{bmatrix}, \notag \
    \,N= \begin{bmatrix} 0 \\ 0 \\ 0 \end{bmatrix} \\
    M_x=\begin{bmatrix}
     -0.03 &   0  &  0\\
   0 &  -0.02  &       0\\
    0.0133 &   0.0134 &  -0.0234
\end{bmatrix}, \notag \
    \,M_w=\begin{bmatrix}
   1       &    0   &  0 \\
    0&  1&    0 \\
    0       &    0   &  1
\end{bmatrix}.
\end{gather*}
and similar to the earlier CT example, $M_u$ and $M_v$ are matrices of zeros.
}

\begin{figure}[t!] 
\centering
{\includegraphics[width=0.48\columnwidth]{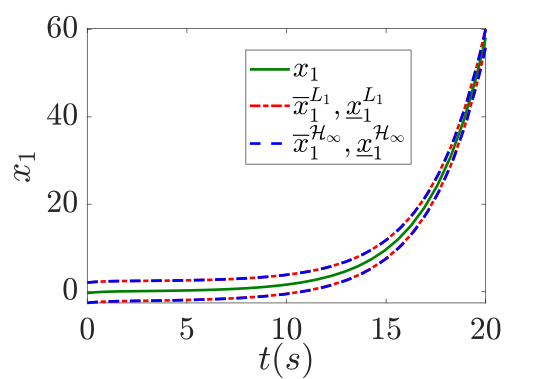}} \label{fig:sub1} 
{\includegraphics[width=0.48\columnwidth]{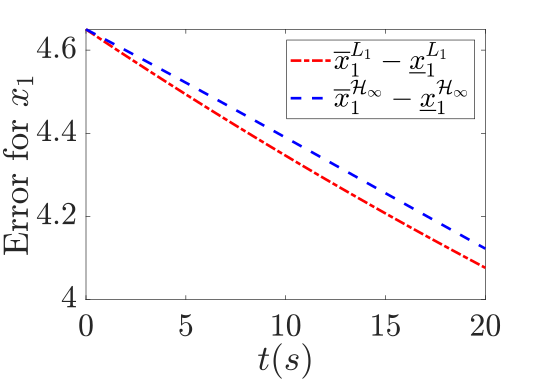}} \label{fig:sub2}
\caption{\small CT Example 2: State, $x_1$, and its upper and lower framers and error \mk{of our proposed observer}, $\overline{x}^{L_1}_3,\underline{x}^{L_1}_3,\varepsilon^{L_1}_3$ for the $L_1$-robust interval observer, \mk{ $\overline{x}^{\mathcal{H}^{\infty}}_3,\underline{x}^{\mathcal{H}^{\infty}}_3,\varepsilon^{\mathcal{H}^{\infty}}_3$ for the $\mathcal{H}^{\infty}$}-robust interval observer.}
\label{fig:figure_ct2}
\end{figure}
\tpa{From Figure \ref{fig:figure_ct2} (for brevity only $x_1$ results are shown, other states follow the same trend), both of the proposed observers have comparable performance in framing the true state with the $L_1$ robust observers error converging slightly faster.}

 \subsection{DT System Example 1}\label{sec:DT_exm}
For the DT case, consider the DT H\'enon chaos system
 \cite[Section V-B]{khajenejad_H_inf_2022}:
\begin{align}
\label{eq:exampletwo}   
x_{t+1} =  Ax_t + r[1 -x_{t,1}^2 ]+Bw_t, \quad y_t =  x_{t,1}+v_t,
\end{align}
with 
$A =
\begin{bmatrix}
0 & 1 \\
0.3 & 0 
\end{bmatrix}$, $B=I$, $r =\begin{bmatrix}0.05 \\ 0\end{bmatrix}$, $\mathcal{W}=[-0.01,0.01]^2$, $\mathcal{V}=[-0.1,0.1]$ and $\mathcal{X}_0 = [-2, 2] \times [-1, 1]$. \tpa{From \eqref{eq:exampletwo}, the nonlinear component of the dynamics is $f(x)=r[1-x_{t,1}^2]$ and its corresponding Jacobian sign-stable mapping is $\phi(x)=f(x)-Hx$, where $H =
\begin{bmatrix}
0.01 & 0\\
0 & 0 
\end{bmatrix}$ is computed using Proposition \ref{prop:JSS_decomp}. Consequently, $\phi_2(x)=0$ and for $\phi_1(x)$, we can construct a tight decomposition function $\phi_{d,1}(x_1,x_2)=\phi_1(D_1x_1+(I-D_1)x_2)$ using Proposition \ref{prop:tight_decomp}  with $D_1=\begin{bmatrix}
   0  & 0   \\
    0  & 0  
\end{bmatrix}$, and its corresponding $F_\phi=\begin{bmatrix}
   0.11  &0   \\
    0  &0   
\end{bmatrix}$ as shown below \eqref{eq:JSS_up_bound}. The solution of \eqref{eq:MILP} then returns the $L_1$-robust observer gains as:
\begin{gather*}
T=\begin{bmatrix}
     1   &   0  \\
   0   &    1
\end{bmatrix}, \notag \
    \,L= \begin{bmatrix}  0\\0.1
 \end{bmatrix}, \notag \
    \,N= \begin{bmatrix} 0\\0\end{bmatrix}, \\M_x= \begin{bmatrix}  0  &  1\\
    0.2&         0\end{bmatrix}, \
    \,M_w= \begin{bmatrix}   1&0\\
    0&         1\end{bmatrix},\notag \
    \,M_v= \begin{bmatrix} 0\\
    0.1\end{bmatrix}.
\end{gather*}
Similarly, the solution of \eqref{eq:MISDP} returns the $\mathcal{H}_\infty$-robust observer gains as:
\begin{gather*}
T=\begin{bmatrix}
     1   &   0  \\
   0   &    1
\end{bmatrix}, \notag \
    \,L= \begin{bmatrix}  0\\0.0304
 \end{bmatrix}, \notag \
    \,N= \begin{bmatrix} 0\\0\end{bmatrix}, \\M_x= \begin{bmatrix}  0  &  1\\
    0.2696 &       0\end{bmatrix}, \
    \,M_w= \begin{bmatrix}   1 & 0 \\
    0 &         1 \end{bmatrix},\notag \
    \,M_v= \begin{bmatrix} 0\\
    0.0304\end{bmatrix}.
\end{gather*}
In both cases, $M_u$ is a zero matrix.
}

\begin{figure}[t] 
\centering
{\includegraphics[width=0.49\columnwidth]{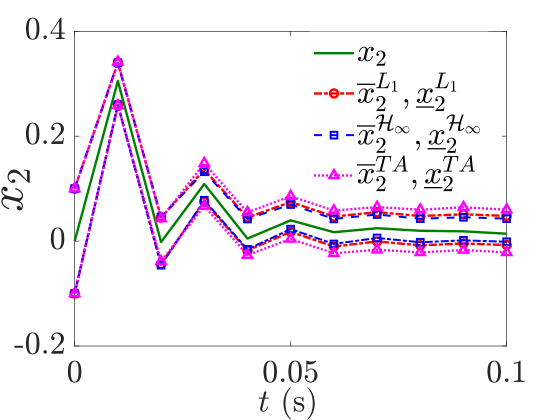}} \label{fig:sub3} 
{\includegraphics[width=0.49\columnwidth]{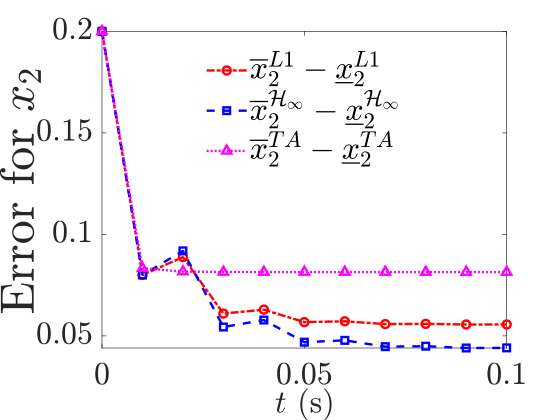}} \label{fig:sub4}
\caption{{{\small DT Example 1: State, $x_2$, and its upper and lower framers and error \mk{of our proposed observers}, $\overline{x}^{L_1}_2,\underline{x}^{L_1}_2,\varepsilon^{L_1}_2$ for the ${L_1}$-robust interval observer, \mk{ $\overline{x}^{\mathcal{H}_{\infty}}_2,\underline{x}^{\mathcal{H}_{\infty}}_2,\varepsilon^{\mathcal{H}_{\infty}}_2$ for the $\mathcal{H}_{\infty}$-robust interval observer, and $\overline{x}^{TA}_2,\underline{x}^{TA}_2,\varepsilon^{TA}_2$} for the observer in \cite{tahir2021synthesis}.}}}
\label{fig:figure2}
\end{figure}

\tpa{From Figure \ref{fig:figure2} (only $x_2$ is considered for brevity, $x_1$ shows a similar trend), the estimates from both our proposed approaches are tighter than the method in \cite{tahir2021synthesis} and when compared to each other, the framers for the $\mathcal{H}_\infty$-robust observer are slightly tighter than the $L_1$-robust version.}\\
\subsection{DT System Example 2}\label{sec:DT_exm}
{For the second discrete time example, consider the slightly modified version of the discrete-time predator-prey system from}
 \cite[Section VB]{khajenejad2021intervalACC}: 
 \begin{align}
\label{eq:exampleDTtwo}   
\begin{bmatrix}
    x_{1,t+1}\\x_{2,t+1}\\d_{t+1}
\end{bmatrix} &=  \begin{bmatrix}
    x_{1,t}\\x_{2,t}\\d_{t}
\end{bmatrix} + h\begin{bmatrix}
    -x_{1,t}x_{2,t}-x_{2,t}+d_t+w_{1,t}\\x_{1,t}x_{2,t}+x_{1,t}+w_{2,t}\\100(\cos(x_{1,t})-\sin(x_{2,t}))+w_{3,t}
\end{bmatrix}, \\ y_t &=  \begin{bmatrix}
    x_{1,t}+v_{1,t}\\x_{2,t}+v_{2,t}\\\sin(d_{t})+v_{3,t}
\end{bmatrix},
\end{align}
with bounded disturbances $w_{l,t}\in[-0.1,0.1],\,\forall l=1,2,3$, bounded noise $v_{l,t}\in[-0.01,0.01],\,\forall l=1,2,3$, sampling time $h=0.05$ and $\mathcal{X}_0 = [-0.35, 0] \times [-0.1, 0.6]\times[-10,10]$.

From \eqref{eq:exampleDTtwo}, the nonlinear component of the dynamics is $$f(x)=\begin{bmatrix}
    
-x_{1,t}x_{2,t}\\ -x_{1,t}x_{2,t}\\100(\cos(x_{1,t})-\sin(x_{2,t}))\end{bmatrix} $$ with the corresponding Jacobian sign-stable mapping $\phi(x)=f(x)-Hx$, where 
$H =
{\begin{bmatrix}
0.001&0 & 0\\
-0.001&0 & 0\\
0&0&0
\end{bmatrix}}$ 
is computed using Proposition \ref{prop:JSS_decomp}. Therefore, using Proposition \ref{prop:tight_decomp}, we can construct a decomposition function $\phi_{d,i}(x_1,x_2)=\phi_i(D^f_ix_1+(I-D^f_i)x_2)$ for each JSS function component $\phi_i(x)$ where $D^f_i=\diag(D^f_{x,i})$ with $D^f_{x,i}$ representing the $i^\text{th}$ row of $D^f_x=
{\begin{bmatrix}
0&1 & 0\\
1&0 & 0\\
1&0&0
\end{bmatrix}}$. Further, the corresponding $F_{\phi}=\begin{bmatrix}
    0.0070 &  0.0035   &   0\\
    0.0070   & 0.0035    &   0\\
    0.3429  &  1.0000     &    0
\end{bmatrix}$ can be computed as shown below \eqref{eq:JSS_up_bound}.

Similarly, the nonlinear component of measurement $h(x)=[0,0,\sin(d_t)]^\top$ and its corresponding JSS function is $\psi(x)=h(x)-Cx$ with $C=
{\begin{bmatrix}
0&0 & 0\\
0&0 & 0\\
0&0&1
\end{bmatrix}}$. Consequently, using
 Proposition \ref{prop:tight_decomp}, $\psi_{d,3}(x_1,x_2)=\psi_3(D^h_3x_1+(I-D^h_3)x_2)$, where $D^h_3=
{\begin{bmatrix}
0&0 & 0\\
0&0 & 0\\
0&0&0
\end{bmatrix}}$ and $F_\psi=\begin{bmatrix}
0&0 & 0\\
0&0 & 0\\
0&0&2
\end{bmatrix}$ can be computed as shown below \eqref{eq:JSS_up_bound}.

Further, $\psi^+(x,w)=\sin(d_{t+1})-d_{t+1}$ and $\rho(x,w)=\psi^+(x,w,u)-A_2x-W_2w$ with 
\begin{gather*}
    A_2=\begin{bmatrix}
       0   &  0   &   0\\
         0  &  0   &   0\\
    0.3429 &  -0.1747    &0 
    \end{bmatrix},\notag \
    \quad W_2= \begin{bmatrix}
       0   &  0   &   0\\
         0  &  0   &   0\\
    0 &  0    &0 
    \end{bmatrix}.
\end{gather*}

As a consequence, the $L_1$-robust observer gains can be computed by solving the optimization problem in \eqref{eq:MILP} to obtain
\begin{gather*}
    T=\begin{bmatrix}
       0   &  0   &   0\\
         0  &  0   &   0\\
    -20 & 17.837 &  1 
    \end{bmatrix},\notag \
    \quad L= \begin{bmatrix}
       0   &  0   &   0\\
         0  &  0   &   0\\
    0 &  0    &0 
    \end{bmatrix},\notag \
    \quad N= \begin{bmatrix}
       1   &  0   &   0\\
         0  &  1   &   0\\
     20&  -17.837   &  0
    \end{bmatrix}
    \end{gather*}
    and the corresponding observer matrices
    \begin{gather*}
     M_x=\begin{bmatrix}
       0   &  0   &   0\\
         0  &  0   &   0\\
     -19.2973&18.837 & 0
    \end{bmatrix},\notag \
    \, M_w= \begin{bmatrix}
       0   &  0   &   0\\
         0  &  0   &   0\\
     -1&0.8918 & 0.05
    \end{bmatrix},\notag \\ M_v= \begin{bmatrix}
      0   &  0   &   0\\
         0  &  0   &   0\\
      -19.2973& 18.837   &  0
    \end{bmatrix}.
    \end{gather*}
Similarly, the $\mathcal{H}_\infty$-robust observer gains computed by solving \eqref{eq:MISDP} are
\begin{gather*}
    T=\begin{bmatrix}
       0   &  0   &   0\\
         0  &  0   &   0\\
    -19.9999  &  0  &  1 
    \end{bmatrix},\notag \
    \quad L= \begin{bmatrix}
       0   &  0   &   0\\
         0  &  0   &   0\\
    -20.0999 &  1    &0 
    \end{bmatrix},\notag \
    \quad N= \begin{bmatrix}
       1   &  0   &   0\\
         0  &  1   &   0\\
     19.9999 & 0  &  0
    \end{bmatrix}
    \end{gather*}
    and the corresponding observer matrices are
    \begin{gather*}
     M_x=\begin{bmatrix}
       0   &  0   &   0\\
         0  &  0   &   0\\
     0 &0 & 0.7246
    \end{bmatrix}\times10^{-5},\notag \
    \, M_w= \begin{bmatrix}
       0   &  0   &   0\\
         0  &  0   &   0\\
    -1&    0  &  0.05 
    \end{bmatrix},\notag \
    \, M_v= \begin{bmatrix}
      0   &  0   &   0\\
         0  &  0   &   0\\
       -20.0997 &  1    &  0
    \end{bmatrix}.
    \end{gather*}
\begin{figure}[t] 
\centering
{\includegraphics[width=0.49\columnwidth]{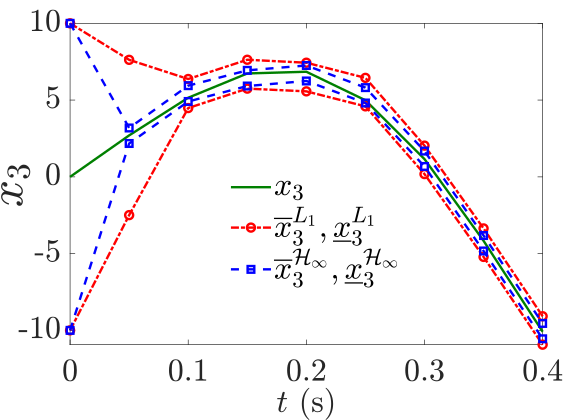}} \label{fig:sub3} 
{\includegraphics[width=0.49\columnwidth]{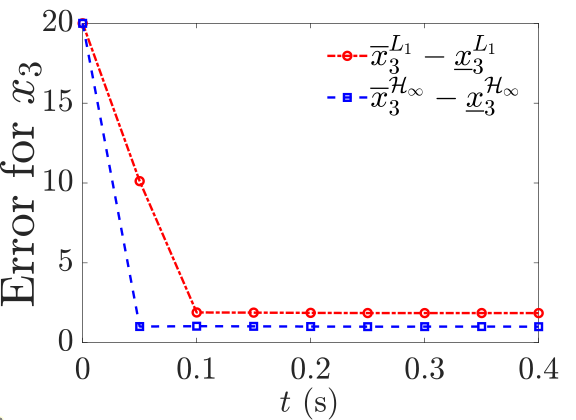}} \label{fig:sub4}
\caption{{{\small DT Example 2: State, $x_3$, and its upper and lower framers and error \mk{of our proposed observers}, $\overline{x}^{L_1}_3,\underline{x}^{L_1}_3,\varepsilon^{L_1}_3$ for the ${L_1}$-robust interval observer, \mk{ $\overline{x}^{\mathcal{H}_{\infty}}_3,\underline{x}^{\mathcal{H}_{\infty}}_3,\varepsilon^{\mathcal{H}_{\infty}}_3$ for the $\mathcal{H}_{\infty}$-robust interval observer.}}}}
\label{fig:figure3}
\end{figure}
 

\section{Conclusion and {Future Work}} 

In this chapter, we introduced a unified framework for two novel interval observer designs for uncertain discrete- and continuous-time nonlinear systems under bounded Jacobians and bounded uncertainty assumptions. Specifically, the proposed observers are positive and correct by design without the requirement for extra positivity constraints or assumptions. This is achieved by leveraging bounding techniques from mixed-monotone decompositions and embedding systems. Moreover, we proved that the proposed observers are input-to-state stable and minimize the $\mathcal{H}_{\infty}$- and $L_1$-gains of the framer error's linear comparison 
systems, respectively. Our proposed approaches involve mixed-integer linear and semi-definite programs, respectively, and both methods offer additional degrees of freedom. These approaches outperform state-of-the-art approaches while having comparable performance to each other,  with the main advantage of the $L_1$-robust interval observer design involving MILPs over the $\mathcal{H}_{\infty}$-robust interval observer design involving MISDP being the availability of faster solvers. As for future extensions, we will consider extensions of our interval observer design framework to uncertain hybrid systems as well as systems with unknown inputs for obtaining resilient state observers.

\bibliography{biblio}
\end{document}